\documentclass[reqno,11pt]{amsart}
\usepackage{changebar}
\usepackage{color}

\usepackage{amssymb,amsthm}

\newcommand{\red}[1]{{\color{red}{#1}}}
\newcommand{\blue}[1]{{\color{blue}{#1}}}

\usepackage{a4}
\numberwithin{equation}{section}

\usepackage[linktocpage]{hyperref}
\usepackage[mathscr]{eucal}

\usepackage{version}
\excludeversion{LONG}

\newcounter{mnotecount}[section]
\renewcommand{\themnotecount}{\thesection.\arabic{mnotecount}}
\newcommand{\mnote}[1]{\protect{\stepcounter{mnotecount}}${\raisebox{0.5\baselineskip}[0pt]{\makebox[0pt][c]{\color{magenta}{\tiny\em$\bullet$\themnotecount}}}}$\marginpar{\raggedright\tiny\em$\!\!\!\!\!\!\,\bullet$\themnotecount: #1}\ignorespaces}

\newcommand{\mymnote}[1]{\mnote{\blue{#1}}}
\renewcommand{\mymnote}[1]{}

\renewcommand{\Re}{\mathbb R}
\newcommand{\tr}{\text{tr}}
\newcommand{\half}{\frac{1}{2}}
\newcommand{\la}{\langle}
\newcommand{\ra}{\rangle}
\newcommand{\eps}{\epsilon}

\newcommand{\Bo}{\mathcal B}
\newcommand{\Sp}{\mathcal S}
\newcommand{\id}{\mathbf{id}}

\newcommand{\LL}{\mathcal L}

\newcommand{\FF}{\mathcal F}

\newcommand{\PP}{\mathbb P}

\newcommand{\bodymetric}{b}
\newcommand{\spacemetric}{g}

\newcommand{\mr}{\mathring}

\newcommand{\pfluid}{p}

\newcommand{\Mass}{\mathcal M}

\theoremstyle{plain}
\newtheorem{thm}{Theorem}[section]

\newtheorem{lemma}[thm]{Lemma}
\newtheorem{example}[thm]{Example}
\newtheorem{definition}[thm]{Definition}
\newtheorem{prop}[thm]{Proposition}
\newtheorem{remark}{Remark}[section]

\title{Elastic deformations of compact stars}
\author[L. Andersson]{Lars Andersson${}^\dagger$} \email{laan@aei.mpg.de}
\address{Albert Einstein Institute, Am M\"uhlenberg 1, D-14476 Potsdam,
  Germany}

\author[R. Beig]{Robert Beig${}^{\ddagger}$
} \email{robert.beig@univie.ac.at}
\address{Gravitational Physics, Faculty of Physics, University of
Vienna,
\newline
Boltzmanngasse 5, A-1090 Vienna, Austria}
\thanks{${}^{\ddagger}$ Supported in part by Fonds zur F\"orderung der
Wissenschaftlichen Forschung project no. P20414-N16.}
\author[B. Schmidt]{Bernd G. Schmidt} \email{bernd@aei.mpg.de}
\address{Max-Planck-Institut f\"ur Gravitationsphysik,
Albert-Einstein-Institut,
\newline
Am M\"uhlenberg 1, D-14476 Golm, Germany}
\begin{document}
\date{February 26, 2014}
\begin{abstract}
We prove existence of solutions for an elastic body interacting with itself through its Newtonian
gravitational field. Our construction works for configurations near one given by a self-gravitating ball of perfect fluid. We use an
implicit function argument. In so doing we have to revisit some classical work in the astrophysical literature
concerning linear stability of perfect fluid stars. The results presented here  extend previous work by the authors,
which was restricted to the astrophysically insignificant situation of configurations near one of vanishing stress. In particular, ``mountains on neutron stars'', which are made possible by the presence of an elastic
crust in neutron stars, can be treated using the techniques developed here.
\end{abstract}

\maketitle

\section{Introduction} \label{sec:intro}
The matter models commonly used to describe stars in astrophysics  are fluids.
By a classical result of Lichtenstein \cite{lichtenstein},
a self-gravitating time-independent fluid body is spherically symmetric in the absence of rotation. However, due to the presence of a solid crust, it is possible for neutron stars to carry mountains (see \cite{crusts}). In view of the result of Lichtenstein, in order to describe something like ``mountains on a neutron star'' it is necessary to consider matter models which allow non-isotropic stresses,
for example those of elasticity theory.

The currently known existence results for non-spherical self-gravitating time-indepent elastic bodies deal with deformations of a relaxed stress-free state. In  \cite{beig:schmidt:2003} it was shown that for a small body for which a relaxed stress-free configuration exists, there is a unique nearby solution which describes the deformation of the body under its own Newtonian gravitational field. The analogous result in Einstein gravity was given in \cite{ABS}. However, a large self-gravitating object like a neutron stars does not admit a relaxed configuration. This is the situation which we shall consider in this paper.

In section \ref{sec:prel}, we give an outline of the formalism of elasticity used in this paper, following \cite{AOS:2011CQGra..28w5006A}, and specializing to the time-independent case. Given a reference body,  i.e., a domain
in three-dimensional Euclidean
body space
the physical body is its image
in
physical space
under a
deformation map.
The field theoretic description of the material in terms of deformation maps gives the material frame (or Lagrangian) form of the theory, while the description in terms of the inverses of the deformation maps, called
configuration
maps,
taking the physical body to the reference body, gives the Eulerian picture. The field equations for Newtonian elasticity can be derived from an action principle  in which the properties of the matter are defined by a ``stored energy function''
which depends on the configuration, and a Riemannian reference metric on the body, called the body metric. The body metric
is taken to be conformal to the Euclidean metric with the conformal factor chosen so that the associated volume element is that given by the particle number density of the reference configuration, see below.
This is done in section \ref{sec:newt}, working in the Eulerian setting. We consider here only isotropic materials, for which the stored energy depends on the configuration map only via the three principal invariants of the strain tensor\footnote{The strain tensor is the push-forward of the contravariant form of the metric on space, which in the present setting can be taken to be the Euclidean metric.} defined with respect to the body metric.

The resulting Euler-Lagrange equations for a self-gravitating elastic body form a quasi-linear, integro-differential boundary value problem for the configuration map, with an a priori unknown boundary given by the condition that
the normal stress be zero there.
In section \ref{sec:materialframe}, we reformulate the system in the material frame, as a system of equations on the reference body, i.e. a fixed domain in Euclidean space. The system of equations on the body given in section \ref{sec:materialframe}  replaces the free boundary of the physical body, by the fixed boundary of the reference body, with a Neumann type boundary condition. For a large class of stored energy functions we obtain by this procedure a quasilinear elliptic, integro-differential boundary value problem, the solutions of which represent self-gravitating time-independent elastic bodies.

There is no general existence theorem in the literature which can be applied to such systems, except in the spherically symmetric case.
A minimizer for the variational problem describing a Newtonian self-gravitating body has been shown to exist, under certain conditions, in  \cite{calogero:leonori:2012arXiv1208.1792C}.
However, it is unknown under what conditions this minimizer satisfies the Euler-Lagrange equations. 
Specializing to the spherically symmetric case, the field equations for a self-gravitating, time-independent body reduces to a system of ordinary differential equations, for which existence of solutions is easier to show. For example, in the case of perfect fluid matter, which we will use as background, as explained below,
there is a simple condition on the equation of state (see \cite{rendall:schmidt:1991CQGra...8..985R} for the relativistic case,  the Newtonian case is analogous) guaranteeing existence of solutions corresponding to a finite body. We shall in this paper construct ``large'' non-spherical time-independent self-gravitating bodies by applying the implicit function theorem to deform appropriate spherical background solutions.

We choose our reference configuration as follows, see sections \ref{sec:refconfig}, \ref{sec:statfluid} for details. As mentioned above, a Newtonian self-gravitating fluid body is spherically symmetric. In fact, a fluid body centered at the origin is determined by specifying its stored energy $e(\rho)$ (in the context of fluids usually referred to as internal energy) as a function of the physical density $\rho$ of the fluid, together with its central density, see again \cite{rendall:schmidt:1991CQGra...8..985R}.
The fluid density $\rho$ can be taken as one of the principal invariants of the strain tensor.

Keeping these facts in mind, a time-independent self-gravitating fluid body with support of finite radius, and non-vanishing density at the boundary, can be viewed as a special case of an elastic material, with deformation map given by the identity, and with spatial density $\mr{\rho}$. A body metric 
conformal to the Euclidean metric
is then chosen, which has volume element
equal to $\mathring\rho$. The two remaining principal invariants of the strain tensor are represented in terms of expressions $\tau, \delta$ which vanish at this background configuration, cf. \cite{shadi:elast} for related definitions.  We then choose a stored energy function
which coincides with $e(\rho)$ at the reference configuration, but which has elastic stresses for non-spherical configurations.
The just described elastic body is then taken as the starting point of an implicit function theorem argument.

Given the stored energy as a function of the principal invariants, the action and hence also the Euler-Lagrange equations are specified. We introduce a 1-parameter family of non-spherical body densities $\mr{\rho}_\lambda$ with $\mr{\rho}_0 = \mr{\rho}$, the spherically symmetric density of the reference body.
Introducing a suitable setting in terms of Sobolev spaces, see section \ref{sec:refbodydeform},
we can formulate the system of Euler-Lagrange equations in the material frame as a 1-parameter family of non-linear functional equations
$\FF[\mr{\rho}_\lambda, \phi_\lambda] = 0$.
Our problem is then reduced to showing that deformation maps $\phi_\lambda$ solving this equation exist for small $\lambda \ne 0$.

In order to apply the implicit function theorem to construct such solutions, we must analyze the Frechet derivative of the $\mathcal F$ or equivalently the equations linearized at the reference deformation. This is done in section \ref{sec:linop}, which forms the core of the paper. We start by calculating the linearized operator. The result is \eqref{eq:lambda-eqs}. With our choice of stored energy, the first variation of the action is ``pure fluid'' (which is why
the identity map is a solution of the Euler-Lagrange equations at the reference configuration), but the second variation, i.e. the Hessian of the action functional, evaluated at the reference configuration, or put differently, the linearization at the identity of the Euler-Lagrange system of equations,
does have an elastic contribution.
Proposition \ref{prop:alternative} gives an alternative way of writing the linearized operator which will be essential later.

On geometrical grounds the kernel of the linearized operator contains the Killing vectors of flat space, translations and rotations. For the application of the implicit function it is essential to show that there are no further elements in the kernel. This is done for radial and
nonradial perturbations separately.

Lemma \ref{lem:radial} gives sufficient conditions for the kernel to contain only the trivial radial perturbation. This is the
famous $3 \gamma-4>0$ condition, which guarantees the absence of linear time-harmonic radial fluid perturbations which grow
exponentially (see \cite{stellar} and historical references therein). Here $\gamma$ is the adiabatic index.
For nonradial perturbations proposition \ref{prop:nonradial} shows that there are just Killing vectors in the kernel. For fluids with vanishing boundary density the analog of this statement is known in the astrophysical literature as the Antonov--Lebovitz Theorem (see \cite{B&T:1987gady.book.....B} and references therein).
Our result shows in particular  that this theorem is also true for fluids with positive boundary density.
Theorem 5.6 states our result that the kernel of the linearized operator contains only Euclidean motions.

To use the implicit function theorem we have also to control the  cokernel. This is the content of proposition \ref{prop:fredholm} which gives a Fredholm alternative for the integro-differential boundary value problem under consideration. Here, an analysis of the regularity properties of the Newtonian potential, and the layer potential are needed. These results are given in appendix \ref{sec:newtest}. Again, it turns out that the cokernel consists of infinitesimal Euclidean motions.

Now all the necessary pieces are in place for applying the implicit function theorem in order to show the existence of a family of solutions in section \ref{sec:implicit}. First we define, following the approach used in \cite{ABS}
the ``projected equations'' by fixing some element of the kernel and projecting on a complement of the cokernel. The implicit function theorem then applies, and we obtain solutions of the projected equations, cf.  proposition \ref{prop:implicit}. However, as is shown in theorem \ref{thm:mainimplicit},
it turns out that these solutions are actually solutions of the full system. The essential reason this holds
is Newton's principle of ``actio est reactio'' i.e. the fact that the self-force and self-torque acting on a body through its own gravitational field is zero, which corresponds to (a nonlinear version of) the condition that the force lie in a complement of the above cokernel.

Finally, in section \ref{sec:physint} we demonstrate that if the family $\mathring\rho_\lambda $ is non-spherical for $\lambda>0$, and if the elastic part in the stored energy is non vanishing, the solutions constructed are not spherically symmetric for $\lambda>0$.
Our analysis gives a clear interpretation of the solutions which have been constructed. They are near a particular self-gravitating time-independent fluid body, and have a small deviation from spherical symmetry which is due to the presence of the elastic terms in stored energy function, and hence also in the system of Euler-Lagrange equations under consideration.

\section{Preliminaries} \label{sec:prel}
\newcommand{\rhonull}{\mathring{\rho}}
\newcommand{\epsnull}{\mathring{\eps}}
\newcommand{\specificmass}{m}
We use the setup of \cite{AOS:2011CQGra..28w5006A},
specializing to the
time independent case, with some changes which shall be specified below.
Let $\Re^3_\Bo, \Re^3_\Sp$ denote the body space and space, respectively. We
introduce Cartesian coordinates $X^A$ on $\Re^3_\Bo$ and $x^i$ on
$\Re^3_\Sp$. We will often write the coordinate derivative operators in
$\Re^3_\Bo$ and $\Re^3_\Sp$ as $\partial_A$ and $\partial_i$, respectively,
and also use the notations $u_{,A} = \partial_A u$, and $u_{,i} = \partial_i
u$.
The body $\Bo$ is an open domain in $\Re^3_\Bo$ with closure $\overline{\Bo}$
and boundary $\partial\Bo$. The physical state of
the material is described by configuration $f: \Re^3_\Sp \to \Bo$ and
deformation $\phi: \Bo \to \Re^3_\Sp$ maps. We assume $f \circ \phi =
\id_\Bo$. The gradients of $f, \phi$ are
denoted
$$
f^A{}_i = \partial_i f^A, \quad \phi^i{}_A = \partial_A \phi^i ,
$$
and satisfy
$$
\phi^i{}_A f^A{}_k = \delta^i{}_k .
$$
The body and space are endowed with Riemannian metrics
$\bodymetric_{AB}$ and $\spacemetric_{ij}$.
\begin{remark}
In the Newtonian case we are considering here, we take the space metric to be
the flat, Euclidean metric $\spacemetric_{ij} = \delta_{ij}$. Further, as will
be explained below, we shall take the metric $\bodymetric_{AB}$ on $\Bo$ to
be conformal to the Euclidean metric.
However, in order to have a clear view of the foundations, it is
convenient to allow these for the moment to be general metrics.
\end{remark}
\newcommand{\bodyvol}{V_\bodymetric}
\newcommand{\spacevol}{V_\spacemetric}
We let $\bodyvol$ be the volume element of $\bodymetric$,
$$
(\bodyvol)_{ABC} dX^A dX^B dX^C = \sqrt{\bodymetric}\, \varepsilon_{ABC} dX^A
dX^B dX^C .
$$
Here $\sqrt{\bodymetric} = \sqrt{\det \bodymetric_{AB}}$ and
$\varepsilon_{ABC}$ is the Levi-Civita symbol.
Similarly let $\spacevol$ denote the volume element of $\spacemetric$.
The number density $n$ is defined
by
\begin{equation}\label{eq:ndef}
f^A{}_{i} f^B{}_{j} f^C{}_{k} (\bodyvol)_{ABC}(f(x)) = n(x) (\spacevol)_{ijk} .
\end{equation}
Let
$$
H^{AB} = f^A{}_i f^B{}_j \spacemetric^{ij} .
$$
Making use of the body metric $\bodymetric_{AB}$ we define
$$
H^A{}_B = H^{AC} \bodymetric_{CB} .
$$
Note that unless explicitly specified, we do not raise and lower indices
with $\bodymetric_{AB}$.
\begin{remark} \label{rem:coordinv}
The eigenvalues of $H^A{}_B$ transform as scalars
with respect to coordinate changes both in $\Re^3_{\Bo}$ and in
$\Re^3_{\Sp}$.
For transformations
of $\Re^3_{\Bo}$ this follows from the fact that $H^A{}_B$ undergoes a
similarity transformation under pullback,
$$
(\zeta^* H)^A{}_B = (\zeta^{-1})^A{}_{,C} H^C{}_D \zeta^D{}_{,B} ,
$$
where $\zeta: \Re^3_{\Bo} \to \Re^3_{\Bo}$. On the other hand, for a
map $k: \Re^3_\Sp \to \Re^3_\Sp$ it holds that the components of $H^A{}_B$
transform as scalars, i.e.,
$$
H^A{}_B[f \circ k, k^* \spacemetric](x) = (H^A{}_B)(k(x))
$$
\end{remark}
In local coordinates, we
have
$$
\sqrt{\spacemetric} \, n = \sqrt{\bodymetric} \,\det(f^A{}_i) ,
$$
which in view of remark \ref{rem:coordinv} gives
$$
n = (\det H^A{}_B)^{1/2} .
$$
Let
\newcommand{\calH}{\mathcal H}
\begin{equation}\label{eq:calH=nH}
\calH^A{}_B = (\det H^A{}_B)^{-1/3} H^A{}_B = n^{-2/3} H^A{}_B .
\end{equation}
Then $\calH^A{}_B$ depends  on the body metric $\bodymetric_{AB}$ only via its
conformal class which
can be represented by
\newcommand{\confclassbody}{[\bodymetric]}
$$
\confclassbody_{AB} = n^{-2/3} \bodymetric_{AB} .
$$
Further, define
\newcommand{\hatcalH}{\sigma}
\begin{equation}\label{eq:hatcalH=calH}
\hatcalH^A{}_B = \calH^A{}_B - \frac{1}{3} \tr \calH \,\delta^A{}_B .
\end{equation}
Given the conformal class of the body metric, we
can in addition to the invariant $n$ also define the invariants
\begin{equation}\label{eq:one}
\tau =  (\calH^A{}_A - 3) ,
\end{equation}
and
\begin{equation}\label{eq:two}
\delta = \det(\hatcalH^A{}_B) .
\end{equation}
See \cite{shadi:elast} for related invariants.
It follows from the definition that the invariants $\tau, \delta$ depend only on the conformal class of the
body metric $b$, and further, by remark \ref{rem:coordinv} we have that
$n,\tau, \delta$ transform as scalars with respect to coordinate changes both
in $\Bo$ and $\Re^3_{\Sp}$.

\subsection{Kinematic identities} \label{sec:identities}
Given the above definitions, we have
\begin{equation}\label{eq:dndf}
\frac{\partial n}{\partial f^A{}_i} = n \phi^i{}_A, \quad
\frac{\partial n}{\partial \phi^i{}_A} = - n f^A{}_i
.
\end{equation}
Let $J = n^{-1} \circ \phi$. The Piola transform of a vector $y^i$ is $Y^A =
J f^A{}_i y^i \circ \phi$. The Piola identity states
\begin{subequations}\label{eq:piola-div}
\begin{align}
\nabla_A Y^A &= J (\nabla_i y^i)\circ \phi,  \\
\intertext{and similarly,}
\nabla_i y^i &= n (\nabla_A y^A) \circ f .
\end{align}
\end{subequations}
Here $\nabla_A$
denotes the covariant derivative with respect to $\bodymetric_{AB}$. The
covariant derivatives and Piola transform also applies to two-point tensors,
see \cite{hughes:marsden} for details.
We also have the following versions of the
Piola identities
\begin{subequations}\label{eq:piola-fphi}
\begin{align}
\nabla_i (n \phi^i{}_A) &= \frac{1}{\sqrt{g}} \partial_i (\sqrt{g} \,n
\phi^i{}_A) = 0, \label{eq:piola-phi} \\
\nabla_A (n^{-1} f^A{}_i) &= \frac{1}{\sqrt{b}}
\partial_A (\sqrt{b} \, n^{-1} f^A{}_i) = 0  . \label{eq:piola-f}
\end{align}
\end{subequations}
Let $H_{AB}$ be defined by the relation
\begin{equation}\label{eq:H_ABdef}
H^{AC} H_{CB} = \delta^A{}_B .
\end{equation}
A calculation shows
\begin{subequations}\label{eq:dndHeqs}
\begin{align}
\frac{\partial n}{\partial H^{AB}} &= \frac{n}{2} H_{AB}  , \label{eq:dndH} \\
\frac{\partial^2 n}{\partial H^{AB}\partial H^{CD}}
&= \frac{n}{4}\,(H_{AB}H_{CD} -
2 H_{A(C}H_{D)B}) . \label{eq:dndH2}
\end{align}
\end{subequations}
Further, we have
\begin{equation}\label{eq:dtaudH}
\frac{\partial \tau}{\partial H^{AB}} = H_{F(A} \hatcalH^F{}_{B)} = H_{FA} \hatcalH^F{}_{B} ,
\end{equation}
and
\begin{equation}\label{eq:dhatcalHdH}
\frac{\partial \hatcalH^A{}_B}{\partial H^{CD}} =
- \frac{1}{3} H_{CD} \hatcalH^A{}_B + H_{R(C} ( \delta^A{}_{D)}
\calH^R{}_B - \frac{1}{3}\calH^R{}_{D)} \delta^A{}_B ) .
\end{equation}

\begin{example}
Suppose $\phi$ is a conformal map, so that $(\phi^* \spacemetric)_{AB} =
\lambda(X) \bodymetric_{AB}$.
Then $H^A{}_B = \lambda^{-1} \delta^A{}_B$ and
hence $n = \lambda^{-3/2}$. This gives
\begin{align}
\calH^A{}_B &= \delta^A{}_B , \\
\confclassbody_{AB} &= (\phi^* \spacemetric)_{AB} .
\end{align}
Further, in this case, the invariants $\tau, \delta$ vanish\footnote{Conversely, when $\tau = \delta = 0$, the map $\phi$ is conformal.}
$\tau = \delta = 0$, and by \eqref{eq:dtaudH} we have
\begin{subequations}\label{eq:dtau-conf}
\begin{align}
\frac{\partial \tau}{\partial H^{AB}}  \bigg{|}_{\phi \text{ conformal}}
&= 0 , \\
\frac{\partial^2 \tau}{\partial H^{AB} \partial H^{CD}}  \bigg{|}_{\phi
  \text{ conformal}} &= H_{A(C} H_{D)B} - \frac{1}{3} H_{AB} H_{CD} .
\end{align}
\end{subequations}
From the definition of the invariant $\delta$ we have $\delta =
O(|\hatcalH|^3)$ and hence
 in case $\phi$ is
conformal,
\begin{equation}\label{eq:ddeltadH}
\frac{\partial \delta}{\partial H^{AB}} \bigg{|}_{\phi \text{ conformal}} = 0
, \quad \frac{\partial \delta}{\partial H^{AB} \partial H^{CD}}
\bigg{|}_{\phi \text{ conformal}} = 0 .
\end{equation}
\end{example}

\section{Field equations of a Newtonian elastic body} \label{sec:newt}
Let $\chi_{f^{-1}(\Bo)}$ denote the indicator function of the support
$f^{-1}(\Bo)$ of
the physical body,  and let
$$
|\nabla U|_\spacemetric^2 = \partial_i U \partial_j U \spacemetric^{ij} .
$$
The field equations for the static
Newtonian self-gravitating elastic body are the
Euler-Lagrange equations for an action of the form
\begin{equation}\label{eq:action-euler}
\LL = \int_{\Re^3_\Sp} \Lambda \spacevol ,
\end{equation}
where
$$
\Lambda = ( \frac{|\nabla U|^2}{8\pi G} + \rho  U \chi_{f^{-1}(\Bo)} + n
\eps \chi_{f^{-1}(\Bo)} ) .
$$
Here $\rho = mn$ is the physical density of the material, where $m$ is the
specific mass per particle, and
\begin{equation}\label{eq:epsframeindiff}
\eps = \eps(H^{AB},f)
\end{equation}
is the stored energy function, which describes
the internal energy per particle of the material.
On the other hand, density of internal energy of the material in its physical state
is $n\eps$.
For a stored energy of the form
\eqref{eq:epsframeindiff}, termed frame indifferent, the action is covariant under spatial diffeomorphisms.
The converse also holds, see \cite{beig:schmidt:relativistic:2003CQGra..20..889B}.

Let $\zeta$ denote the fields  $(f^A, U)$.
The canonical stress-energy tensor is
$$
T^i{}_j = \frac{\partial \Lambda}{\partial (\partial_i \zeta)} \partial_j
  \zeta - \Lambda \delta^i{}_j .
$$
Making use of the covariance of the action,
one finds,
cf. \cite{KM}, that the Euler-Lagrange equations
take the form
\begin{equation}\label{eq:conservation}
\nabla_i T^i{}_j = 0 ,
\end{equation}
where $\nabla_i$ denotes the covariant derivative with respect to the metric
$\spacemetric_{ij}$.

Let $\tau^i{}_j$ and $\Theta^i{}_j$ denote the contributions in $T^i{}_j$
from the elastic field $f^A$ and the Newtonian potential $U$, respectively,
so that
$$
T^i{}_j = \tau^i{}_j \chi_{f^{-1}(\Bo)} + \Theta^i{}_j ,
$$
where
$$
\tau^i{}_j = n \frac{\partial \eps}{\partial f^A{}_i} f^A{}_j ,
$$
and
$$
\Theta^i{}_j = \frac{1}{4\pi G} (\nabla^i U \nabla_j U - \half \nabla^k U
\nabla_k U \delta^i{}_j) .
$$
\subsection{Equations in Eulerian form} \label{sec:eulerian}
From the above, we see that
the Eulerian form of the field equations for a self-gravitating
body are
\newcommand{\norm}{n}
\begin{subequations}\label{eq:elast-euler}
\begin{align}
\nabla_i \tau^i{}_j + \rho \nabla_i U &= 0, \quad \text{ in
  $f^{-1}(\Bo)$}, \label{eq:elast-force-euler} \\
\tau^i{}_j \norm_i \big{|}_{\partial f^{-1}(\Bo)} & = 0 ,
 \label{eq:boundarycond-euler} \\
\Delta_\spacemetric U &= 4\pi G \rho
\chi_{f^{-1}(\Bo)} . \label{eq:poisson-euler}
\end{align}
\end{subequations}
Here $n_i$ is the outward pointing normal to the boundary $\partial
f^{-1}(\Bo)$ of the physical body. The boundary condition
\eqref{eq:boundarycond-euler} follows from the conservation equation
\eqref{eq:conservation}, see \cite[Lemma 2.2]{ABS}.
The Newtonian potential $U$ is taken to be the
unique solution to the Poisson
equation \eqref{eq:poisson-euler} such that $U(x) \to 0$, as $x \to \infty$.

\subsection{Integral form of the Newtonian potential}
With $\spacemetric_{ij} = \delta_{ij}$, the solution to
\eqref{eq:poisson-euler} takes the form
$$
U = E\star(4 \pi G \rho \chi_{f^{-1}(\Bo)}) ,
$$
where $E(x) = - 1/(4\pi|x|)$ is the fundamental solution of $\Delta$, i.e.,
$$
U(x) = -G\int_{f^{-1}(B)} \frac{\rho(x')}{|x-x'|} dx' ,
$$
where $dx$ is the Euclidean volume element.
\begin{remark}\label{rem:eulerianaction}
Substituting the solution of the Poisson equation
  \eqref{eq:poisson-euler} into the action \eqref{eq:action-euler}, one finds
after a partial integration
\begin{equation}\label{eq:LL-BT}
\LL = \int_{f^{-1}(\Bo)} (\half \rho U + \rho \eps ) dx\,,
\end{equation}
where $\epsilon$ is given by \eqref{eq:epsform}. The form of the action given in \eqref{eq:LL-BT}, which is expressed in terms of the configuration $f^A$, can be compared with the energy expression discussed in  \cite[Appendix 5.B]{B&T:1987gady.book.....B}.
\end{remark}
Using the differentiation formula for convolutions, we have
$$
\partial_i U = 4\pi G (\partial_i E) \star (\rho \chi_{f^{-1}(\Bo)}) = G E \star
\partial_i (\rho \chi_{f^{-1}(\Bo)}) ,
$$
where in the last equality, the derivative is in the sense of
distributions.
Hence, the force is given by the integral expression
\begin{align}
-\partial_i U(x) &=  G \int_{f^{-1}(B)}
\left ( \partial_i \frac{1}{|z|} \right )
\bigg{|}_{z=x-x'} \rho(x') d x' .
\label{eq:force-int}
\end{align}
\subsection{Equations in material frame} \label{sec:materialframe}
We now write the elastic system in the material frame which is the form which
we shall use.
The change of variables formula applied to \eqref{eq:action-euler} gives
\newcommand{\material}{\text{mat}}
\begin{equation}\label{eq:action-material}
\LL = \int_{\Bo} (\half \mr{\rho} \, U \circ \phi + \mr{\rho} \eps ) dX ,
\end{equation}
where $dX$ is the Euclidean volume element on $\Bo$, and where $\eps$ is taken to be of the form
\begin{equation}\label{eq:eps-material}
\eps = \eps(H^{AB},X) ,
\end{equation}
with $H^{AB}$ the inverse of $\phi^i{}_A \phi^j{}_B \spacemetric_{ij} \circ \phi$.
The Newtonian potential in the matieral frame is of the form
$$
(U \circ \phi)(X) = - G \int_{\Bo} \frac{\rhonull(X') }{|\phi(X) - \phi(X')|} d X' .
$$

Let
\newcommand{\taudens}{\tilde\tau}
\begin{equation}\label{piolacauchy}
\taudens^A{}_i =
- \sqrt{\bodymetric} J (f^A{}_k \tau^k{}_i) \circ \phi .
\end{equation}
Then $\taudens$ is minus the Piola transform of
$\tau^i{}_j$, densitized by the factor $\sqrt{\bodymetric}$ ($= \rhonull$ in
Cartesian coordinates).
We have
\begin{equation}\label{eq:taudens-def}
\taudens^A{}_i = \frac{\partial(\sqrt{\bodymetric}\, \eps)}{\partial \phi^i{}_A}\,,
\end{equation}
where the $\eps$ in \eqref{eq:taudens-def} is understood to be as in \eqref{eq:eps-material}.
The Piola identity implies
$$
\nabla_A \taudens^A{}_i = \sqrt{\bodymetric}\,J (\nabla_i \tau^i{}_j)\circ \phi .
$$
With $\spacemetric_{ij} = \delta_{ij}$ and
$\bodymetric_{AB} = \rhonull^{2/3} \delta_{AB}$, we have $J =
\rhonull^{-1}\det(\phi^i{}_A)$,
$\sqrt{\bodymetric} = \rhonull$. In view of the definition of $\taudens^A{}_i$
we have
$$
\nabla_A \taudens^A{}_i =  \partial_A \taudens^A{}_i ,
$$
and we can write the system
\eqref{eq:elast-euler} in the material frame as
\begin{subequations}\label{eq:elastmaterialsyst}
\begin{align}
- \partial_A (\taudens^A{}_i) +
\rhonull (\partial_i U) \circ \phi &= 0 \quad
   \text{ in } \Bo , \label{eq:elastmaterial}
\\
\taudens^A{}_i n_A|_{\partial\Bo} &=0 ,
\label{eq:elastmaterial-boundary}
\end{align}
\end{subequations}
where $n^A$ is the normal to $\partial \Bo$.  Using the change of variables
formula, and \eqref{eq:force-int}
the gravitational term $\rhonull
(\partial_i U) \circ \phi$ in \eqref{eq:elastmaterial} can be written in the integral form
\begin{equation}\label{eq:force-mat-int}
(\rhonull (\partial_i U) \circ \phi)(X) =
- G \rhonull(X) \int_{\Bo} \left ( \partial_i \frac{1}{|z|} \right ) \bigg{|}_{z
  = \phi(X) - \phi(X')} \rhonull(X') d X' .
\end{equation}

\subsection{The reference configuration} \label{sec:refconfig}
We now introduce the situation which we shall consider in the rest of this
paper. We restrict the space metric to be the flat metric $\spacemetric_{ij}
= \delta_{ij}$ and work in Cartesian coordinates $(X^A)$, $(x^i)$ on the body
and space respectively. The radial coordinates in $\Re^3_\Sp$ and $\Re^3_\Bo$
are
$R = |X| = \sqrt{X^A X^B \delta_{AB} }$ and $r = |x| = \sqrt{x^i x^j
  \delta_{ij}}$.
In the following,
we will set the specific mass $\specificmass =1$ and
denote both the number density and the density of the material by $\rho$.

\begin{definition}\label{def:refconfig}
The reference configuration is given by
choosing the
body domain to be $B(R_0)$, the
ball of radius $R_0$ i.e.
$$
\Bo = \{ |X| < R_0 \} ,
$$
with the trivial configuration and deformation $f = \id$, $\phi=\id$, i.e.
$$
f^A(x) = \delta^A{}_i x^A, \quad A=1,2,3, \qquad \phi^i(X) = \delta^i{}_A X^A, \quad i=1,2,3 .
$$
For a given, positive function $\rhonull$ on $\Bo$ called
the reference density, we set
$$
\bodymetric_{AB} = \rhonull^{2/3} \delta_{AB} .
$$
\end{definition}
In
Cartesian coordinates $(X^A)$, $(x^i)$ we have
\begin{equation}\label{eq:important}
\rho = \rhonull
\det(f^A{}_i)= \mr{\rho} \det(\phi^i{}_A)^{-1}
,
\end{equation}
and the volume element of $\bodymetric_{AB}$ takes the form
$$
\bodyvol = \rhonull(X) dX .
$$
For the reference configuration, we have
$$
\rho(x) = \rhonull(f(x)) ,
$$
and the invariants $\tau, \delta$ vanish,
$$
\tau = \delta = 0 .
$$
We will assume an isotropic
stored energy function, i.e. one depending on configurations only via the invariants
$(\rho,\tau,\delta)$, where $\rho$ is given by \eqref{eq:important} and
and $\tau$, $\delta$ are given by \eqref{eq:one} and \eqref{eq:two}, respectively.
We restrict for simplicity to stored energy functions which have no explicit dependence on the configuration $f$.
Note that $\tau$ and $\delta$ are independent
of $\mr{\rho}$.
By Taylor's formula,
$\eps = \eps(\rho,\tau,\delta)$ can,  for configurations near the reference configuration, be written in the form
\newcommand{\efluid}{e}
\begin{equation}\label{eq:epsform}
\eps(\rho,\tau,\delta) = \efluid(\rho) + \tau l(\rho,\tau,\delta) + \delta m(\rho,\tau,\delta) .
\end{equation}
for some functions $l,m$.

Here we may view
$\efluid(\rho)$ as the stored energy for a fluid configuration. Due
to the fact that the invariants $\tau, \delta$ vanish to first order at the
reference configuration, the Euler-Lagrange equation in the reference state
with $f = \id$, $\phi = \id$ will involve only $\efluid(\rhonull)$.

\subsection{Static self-gravitating fluid bodies} \label{sec:statfluid}
We now consider a fluid with stored energy $\efluid = \efluid(\rho)$ at the
reference configuration $f = \id$.
A calculation shows
$$
\tau^i{}_j = \rho^2  \efluid'(\rho) \delta^i{}_j ,
$$
where $\efluid' = d\efluid/d\rho$.
The pressure of a fluid is given in terms of the energy density $\mu =
\rho\eps$ by
$$
\pfluid = \rho \frac{d\mu}{d\rho} - \mu ,
$$
or
\begin{equation}\label{eq:pfluid-e}
\pfluid(\rho) = \rho^2 e'(\rho) .
\end{equation}
Hence,
$$
\tau^i{}_j = \pfluid \delta^i{}_j ,
$$
and equation \eqref{eq:elast-euler}
takes the form
\begin{subequations}\label{eq:fluid-euler}
\begin{align}
\nabla_i \pfluid + \rho \nabla_i U &= 0  , \\
\pfluid \big{|}_{\partial f^{-1}(\Bo)} &= 0  ,  \label{eq:fluid-euler-bc} \\
\Delta U &= 4\pi G \rho \chi_{f^{-1}(\Bo)} .
\end{align}
\end{subequations}
Static fluid bodies in Newtonian gravity are radially symmetric (see \cite{lichtenstein}),
$\rho = \rho(r)$.
Let
$$
\Mass(r) = 4\pi \int_0^r \rho(s) s^2 ds
$$
be the mass contained within radius $r$, so that $\Mass = \Mass(R_0)$ is the total
mass of the body. It follows from the radial Poisson equation that
$$
\nabla_i U = G \frac{\Mass(r)}{r^2} \frac{x_i}{r} ,
$$
and the self-gravitating fluid is therefore governed by the ordinary
differential equation
\begin{equation}\label{eq:dpdr}
\frac{1}{\rho} \frac{d\pfluid}{dr} + G \frac{\Mass(r)}{r^2} = 0 .
\end{equation}
Let $\rho_c > \rho_0 > 0$ be given and consider an equation of state
$\pfluid = \pfluid(\rho)$,
where $\pfluid$ is a smooth, non-negative function with $\frac{d p}{d \rho} > 0$ in
$[\rho_0,\rho_c]$, with $\pfluid(\rho_0) = 0$.
Given an equation of state $p(\rho)$ satisfying the above properties and given $\rho_c$, there is a unique static fluid body, centered at the origin, which has boundary density $\rho_0$ for some radius $R_0$, see \cite{rendall:schmidt:1991CQGra...8..985R} for details\footnote{This reference treats the relativistic case. The Newtonian case considered here is similar but simpler.}.

The adiabatic index of the fluid is given by
\begin{equation}\label{eq:gammadef}
\gamma = \frac{\rho}{\pfluid}\frac{d\pfluid}{d\rho} .
\end{equation}
The following stability condition on $\gamma$,
$$
3\gamma - 4 > 0,
\quad \text{for $\rho \in (\rho_0,\rho_c]$} ,
$$
plays an important role in our argument.
As a simple explicit example of an equation of state with the above stated
properties, consider
\begin{equation}\label{eq:goodstate}
\pfluid = D (\rho-\rho_0) ,
\end{equation}
for some (dimensional) constant $D>0$.
Then we have
$$
\pfluid(3\gamma-4) = D (-\rho + 4 \rho_0) .
$$
It follows that for the equation of state \eqref{eq:goodstate}, the stability
condition holds for central density $\rho_c$ satisfying $\rho_0 < \rho_c < 4 \rho_0$.

\section{The reference body and its deformation} \label{sec:refbodydeform}
We shall construct a family of static, self-gravitating elastic bodies by
deforming a reference body.
The reference body is a static fluid body in the reference configuration,
i.e. a solution of the equations \eqref{eq:fluid-euler} with
internal energy $e = e(\rho)$, of
radius $R_0$ and density $\rhonull$. As in section \ref{sec:statfluid} we assume that the boundary density is
positive,
$$
\rhonull_0 = \rhonull(R_0) > 0 ,
$$
and that $\mr{p} = \rhonull^2 e'(\rhonull)$, the pressure in the reference
body, satisfy $\mr{p} > 0$ in $\Bo$, $\mr{p} = 0$ if $R=R_0$ and $\frac{d \mathring{p}}{d \mathring{\rho}} > 0$
 in $\bar{\Bo}$. It follows that $\rhonull \geq \rhonull_0 >
0$. Note that the condition $\mr{\rho}(R_0) = 0$ is simply the boundary condition
\eqref{eq:fluid-euler-bc}, which follows from the variational principle.

Given a static self-gravitating
fluid reference body with internal energy function
$e(\rho)$ we consider deformed static self-gravitating bodies with
isotropic stored energy function of the form \eqref{eq:epsform}, with
$\eps(\rho,0,0) = e(\rho)$. Let
$$
\mr{l} = l(\rhonull,0,0) .
$$
We assume
$$
\mr{l}  > 0 ,
$$
in $\overline{\Bo}$.

\subsection{Analytical formulation}
Let
$$\FF: C^\infty(\Bo) \times C^\infty(\Bo; \Re^3_\Sp)
\to C^\infty(\Bo ; \Re^3) \times
C^\infty(\partial\Bo ; \Re^3)
$$
be defined by
\begin{equation}\label{eq:Fdef}
\FF[\rhonull,\phi] = (- \partial_A (\taudens^A{}_i) +
\rhonull (\partial_i U) \circ \phi , \taudens^A{}_i n_A) .
\end{equation}
Let $H^s(\Bo)$ and $H^s(\partial \Bo)$ denote the $L^2$ Sobolev spaces, see
\cite[Chapitre 1]{lions:magenes:vol1} for background.
For $k \geq 1$, $\FF$ extends to a smooth map
$$
\FF: H^{2+k}(\Bo) \to H^k(\Bo) \times H^{1/2+k}(\partial\Bo) .
$$
For simplicity, we assume $k$ is an integer in the following.
\begin{remark} The same statement holds true if we replace the $L^2$ Sobolev
  space by the spaces $W^{2+k,p}(\Bo)$ and the corresponding boundary space
  $B^{k+1-1/p,p}(\partial\Bo)$, with $k \geq 0$, $p > 3$. With $k=0$, this
  corresponds to the situation considered in \cite{ABS}.
\end{remark}
The system \eqref{eq:elastmaterialsyst} for a self-gravitating body
takes the form
$$
\FF[\rhonull,\phi] = 0 ,
$$
which we consider as an equation for the deformation $\phi$ given a
reference density $\rhonull$. Letting $\lambda \mapsto \rhonull_\lambda$ be a
one-parameter family of reference densities, with $\rhonull_0 = \rhonull$,
we shall apply the implicit
function theorem to show that for $\lambda$ sufficiently close to zero, there
is a solution $\phi_\lambda$ of the equation
$$
\FF[\rhonull_\lambda, \phi_\lambda] = 0 .
$$
In order to do this, we must analyze the Frechet derivative
$$
D_\phi \FF[\rhonull, \id] .
$$
This is done in the next section.

\section{The linearized operator} \label{sec:linop}
Let $\xi^i$ be a vector field on $\Re^3_\Sp$ and consider a 1-parameter family
of deformations, $\phi_s$, with $\phi_0 = \id$, such that
$$
\frac{d}{ds} \phi^i_s(X) \bigg{|}_{s=0} = \xi^i(\id(X)) .
$$
The linearization with respect to $\phi$
of the map $\FF$, at the reference
configuration, is given by
$$
D_\phi \FF[\rhonull,\id].\xi = \frac{d}{ds} \bigg{|}_{s=0} \FF[\rhonull,
  \phi_s] .
$$
We define the operator $\xi \mapsto L_i[\xi]$ and
the linearized boundary operator $\xi \mapsto l_i[\xi]$ by
$$
D_\phi \FF[\rhonull,\id].\xi = (-L_i[\xi], l_i[\xi]) .
$$
We now calculate the the explicit form of $L_i[\xi]$ and $l_i[\xi]$.

\newcommand{\g}{$\bf \bar g g^{\alpha\beta}$}
\subsection{Derivation}
\newcommand{\taunull}{\mathring{\tau}}
Let $\taunull^A{}_i$ be defined by
\newcommand{\pnull}{\mathring{p}}
$$
\taunull^A{}_i = \taudens^A{}_i \bigg{|}_{\phi = \id} .
$$
Then
\begin{equation}\label{unperturbed}
\mathring{\tau}^A{}_i = - \mathring{p}\,\delta^A{}_i \, .
\end{equation}
Recall that with $\phi = \id$, we have $H_{AB} = \delta_{AB}$.
A calculation using
the facts recorded in section \ref{sec:identities} gives
\begin{multline}\label{piola}
\frac{d}{ds} \taudens[\rhonull,\phi_s]^A{}_i
\bigg{|}_{s=0}
= [2 \mathring{p}\,\delta^{[A}{}_j \delta^{B]}{}_i
+ \mathring{\rho}\frac{d \mathring{p}}{d \mathring{\rho}}\,\delta^A{}_i
\delta^B{}_j \\
+
4 \mathring{\rho}\,\mathring{l}\,(\delta^{(A}{}_i \delta^{B)}{}_j -
\frac{5}{6}\,\delta^A{}_i \delta^B{}_j + \frac{1}{2}\, \delta^{AB}\delta_{ij})]\,\xi^j{}_{,B} ,
\end{multline}
where $\mathring{l} = l(\mathring{\rho},0,0)$.
With $\phi = \id$, we have the identification $x^i = \phi^i(X)$. We will
sometimes write $X^i$ for $\phi^i(X)$. The chain rule gives
$$
\partial_i
= \delta^A{}_i
\partial_A \,.
$$
and we will often make use of this notation in the following.
Further, define  $(\delta_\xi \mathring{\sigma})_{ij}$ by
$$
(\delta_\xi \mathring{\sigma})_{ij} = \xi_{(i,j)} - \frac{1}{3}\,\delta_{ij}\,
\xi^l{}_{,l} \, ,
$$
so that $2(\delta_\xi \mathring{\sigma})_{ij}$ is
the Euclidean space conformal Killing operator, acting on $\xi$.

The operator $L_i[\xi]$ can now be written in the form
\begin{subequations}\label{eq:lambda-eqs}
\begin{align}
L_i [\xi] &= \partial_i\left[\left(- \mathring{p} + \mathring{\rho} \,\frac{d \mathring{p}}{d \mathring{\rho}}\right)\xi^j{}_{,j}\right] +
\partial_j(\mathring{p}\, \xi^j{}_{,i}) + 4 \,\partial^j
\left[\mathring{\rho}\, \mathring{l} \,(\delta_\xi
  \mathring{\sigma})_{ij}\right] \nonumber \\
&\quad + \,G \int_{\mathcal{B}}\mathring{\rho}(X)\mathring{\rho}(X')\left(\partial_i \partial_j\frac{1}{|X-X'|}\right)
[\xi^j(X) - \xi^j(X')]\,dX'  . \label{lambda5}
\end{align}
Recalling that $\mr{p}\,(R_0)=0$, we find that the linearized
boundary operator is given by
\begin{equation}
l_i[\xi] = \left[\mathring{\rho}\,\frac{d \mathring{p}}{d \mathring{\rho}}\,\xi^j{}_{,j}\,n_i +
4 \,\mathring{\rho}\, \mathring{l}\, (\delta_\xi \mathring{\sigma}_{ij})
n^j\right]
\bigg{|}_{\partial \mathcal{B}}
\label{lambda5'}
\end{equation}
\end{subequations}
Here we have used the notation
$n^i = \frac{1}{R} X^i$.

The operator $L_i$ is formally self-adjoint. This follows from the fact that
it arises from a variational problem, but can easily be checked
explicitely. Let $(\xi, \eta)$ denote the $L^2$ inner product
$$
(\xi, \eta) = \int_\Bo \xi^i \eta_i dX .
$$
The Green's identity for $(-L[\xi], l[\xi])$ takes the form
\begin{equation}\label{eq:Green}
(L[\xi],\eta) - (\xi, L[\eta]) = \int_{\partial\Bo} l_i[\xi] \eta^i dS(X) -
  \int_{\partial\Bo} \xi^i l_i[\eta] dS(X) \, .
\end{equation}
The following proposition gives an alternate form for $L_i[\xi]$ which will
play an important role in the following.
\begin{prop}\label{prop:alternative}
The operator $L_i$ in (\ref{lambda5}) can be written as
\begin{equation}\label{id7}
L_i[\xi] = \mr{\rho}\,\partial_i\left(\frac{1}{\mr{\rho}}\frac{d \mr{p}}{d\mr{\rho}}\, \delta_\xi \mr{\rho} + \delta_\xi \mr{U}\right) +
4 \,\partial^j \left[\mathring{\rho}\, \mathring{l}\,(\delta_\xi \mr{\sigma})_{ij}\right]\,,
\end{equation}
where $\delta_\xi \mathring{\rho} = (\mathring{\rho}\,\xi^i)_{,i}$ and
\begin{align*}
(\delta_\xi \mr{U})(X) &= -  \red{G} \partial_j \int_\mathcal{B}\frac{\mathring{\rho}(R')\,\xi^j(X')}{|X - X'|} dX' \\
&= - G\int_{\mathcal{B}} \frac{(\delta_\xi \mr{\rho}) (X')}{|X-X'|} dX'
+ G \int_{\partial \mathcal{B}}
\frac{\mr{\rho}(R')}{|X-X'|}
\xi^i(X') n_i(X') dS(X') .
\end{align*}
Further, the boundary operator $l_i$ in (\ref{lambda5'}) can be written as
\begin{equation}\label{rewrite1}
l_i[\xi] = \left[\left(\frac{d \mathring{p}}{d \mathring{\rho}}\,\delta_\xi \mr{\rho} + \mr{\rho} \,\mr{U}' \xi^j n_j\right)n_i +
4 \mr{\rho}\,\mr{l}\, (\delta_\xi \mr{\sigma})_{ij} n^j\right]\,\bigg{|}_{\partial \mathcal{B}} .
\end{equation}
\end{prop}
\begin{proof}
Let us call the first, second and fourth term in (\ref{lambda5}) respectively $A$, $B$ and (second line) $C$.
The elastic term in the first line is not affected. We then have that
\begin{align}
C &= - G \mr{\rho}\red{(X)} \partial_i \partial_j \int_{\mathcal{B}}
\frac{\mathring{\rho}(X') \xi^j(X')}{|X - X'|} dX'
- \mathring{\rho}(X) \xi^j(X) \partial_i \partial_j \mathring{U}(X) ,
\label{C}
\intertext{so that}
A + B + C &= - G \rhonull(X) \partial_i \partial_j \int_{\mathcal{B}} \frac{\mathring{\rho}(X') \xi^j(X')}{|X - X'|} dX' +
[-(\partial_i \mathring{p})\xi^j{}_{,j} + (\partial_j \mathring{p})
  \xi^j{}_{,i}] \nonumber \\
&\quad +\partial_i\left(\mathring{\rho} \frac{d \mathring{p}}{d \mathring{\rho}} \xi^j{}_{,j}\right) \red{+} \mathring{\rho}(X) \xi^j(X)
\partial_i \partial_j \mathring{U}(X) \nonumber \\
&= - G \,\partial_i \partial_j \int_{\mathcal{B}}
\frac{\mathring{\rho}(X') \xi^j(X')}{|X - X'|} dX' + \mathring{\rho}
\,\partial_i\left[\frac{d \mathring{p}}{d
    \mathring{\rho}}\frac{(\mathring{\rho}
    \xi^j)_{,j}}{\mathring{\rho}}\right]  . \label{last}
\end{align}
To check the last equality in (\ref{last}), we calculate
\begin{equation}\label{non}
\mathring{\rho} \,\partial_i\left[\frac{d \mathring{p}}{d \mathring{\rho}}\frac{(\mathring{\rho} \,\xi^j)_{,j}}{\mathring{\rho}}\right] =
\mathring{\rho} \,\partial_i\left[\frac{d \mathring{p}}{d \mathring{\rho}}\frac{\mathring{\rho}_j \xi^j}{\mathring{\rho}}\right] +
\partial_i \left[\mathring{\rho}\frac{d \mathring{p}}{d \mathring{\rho}} \xi^j{}_{,j}\right] - \frac{d \mathring{p}}{d\mathring{\rho}}
\xi^j{}_{,j} \partial_i \mathring{\rho} .
\end{equation}
But the first term in (\ref{non}) can be rewritten as
\begin{equation}\label{non1}
\mathring{\rho}\, \partial_i\left[(\partial_j \mathring{p}) \frac{\xi^j}{\mathring{\rho}}\right] = (\partial_j \mathring{p}) \xi^j{}_{,i} + \mathring{\rho}\, \partial_i\left[- \frac{\mathring{\rho}\,\partial_j \mathring{U}}{\mathring{\rho}}\right] \xi^j\,,
\end{equation}
where we have used the $\mathring{p}' = \mathring{\rho} \mathring{U}'$. in the last term. Inserting (\ref{non1}) into (\ref{non}) proves the second equality
in (\ref{last}), which in turn proves (\ref{id7}). The proof of
(\ref{rewrite1}) is straightforward.
\end{proof}
The Euclidean invariance of the action \eqref{eq:action-euler} governing a
self-gravitating body implies that the linearized operator and
the linearized boundary operator annihilate Euclidean Killing vectors.
For translations this follows from \eqref{eq:lambda-eqs},
and for rotations from \eqref{id7}, \eqref{rewrite1}.
We shall show that under suitable conditions on the fluid
equation of state,
these are the only elements in the kernel of the linearized operator
$\xi \mapsto (-L[\xi],l[\xi])$.
To achieve this, it is vital to decompose the perturbations $\xi^i$ into
radial and non-radial parts.

A vector field $\xi^i$ on $\mathcal{B}$ is called radial if it is of the form
$\xi^i(X) = F(R) n^i(X)$, and
non-radial if $\int_{S^2(R)} \xi^i n_i \,d S = 0$ for all $R \in (0,R_0)$.
If we consider conformal Killing vectors on $\Bo$,
we have that
translations, rotations and conformal boosts are non-radial, while dilations are radial.
The radial part of a vector field $\xi$ is given by
\begin{equation}
\xi^i_{r}(X) = \frac{1}{4 \pi R^2} \int_{\partial B(R)}\!\!\!\xi^j(X')n_j(X') dS(X')\, n^i(X)\,.
\end{equation}
where $B(R)$ is the Euclidean ball of radius $R$.
This gives a unique decomposition
$$
\xi = \xi_r + \xi_{nr}
$$
of $\xi$ into a radial and a non-radial part.
\begin{lemma}
Radial and non-radial vector fields are $L^2$-orthogonal in terms of
$( \cdot , \cdot)$. Let $\xi$ be a vector field with radial and non-radial
parts $\xi_r$ and $\xi_{nr}$, respectively. Then $L[\xi_r]$ is radial and
$L[\xi_{nr}]$ is non-radial.  If $l_i[\xi]$ is zero,
then $l_i[\xi_r] = 0$ and $l_i[\xi_{nr}] = 0$.
\end{lemma}
\begin{proof}
The first claim is obvious. For the second claim concerning radial fields one
simply observes radial fields are - as vector fields - invariant under
rotations and that rotations commute with $L_i$, viewed as an operator from
vector fields to covector fields. For non-radial fields we could not find such
a simple conceptual argument, but going through the terms in (\ref{id7}) this
is not difficult to check, e.g. $\delta_\xi \mr{\rho}$ has zero spherical
mean by $\int_{\partial B(R)}
\xi^i n_i dS = 0$ and the Stokes theorem. In a similar way we see that
\begin{equation}\label{see}
\int_{\partial \mathcal{B}} l_i[\xi_{nr}]n^i dS(X) = 0 .
\end{equation}
To prove the last claim in lemma 5.2, we note that
when $l_i[\xi] = 0$, it follows that
\begin{equation}
\int_{\partial \mathcal{B}} l_i[\xi_{r}]n^i dS(X) = - \int_{\partial \mathcal{B}} l_i[\xi_{nr}]n^i dS(X) = 0 .
\end{equation}
But by spherical symmetry $l_i [\xi_{r}]$ is proportional to $n_i$, so that
$l_i [\xi_{r}] = l_i[\xi_r + \xi_{nr}]$ is zero.
\end{proof}

By the previous lemma, we have
\begin{equation}\label{separate}
(\xi, L[\xi]) = (\xi_{r}, L [\xi_{r}]) + (\xi_{nr}, L [\xi_{nr}])
\end{equation}
and if $l_i[\xi] = 0$,
 \begin{equation}\label{both}
(\xi, L[\xi]) = (\xi_{r}, L [\xi_{r}]) + (\xi_{nr}, L [\xi_{nr}]), \quad
\text{ with } l_i[\xi_r] = 0 = l_i[\xi_{nr}] .
\end{equation}

Our next result concerns
the  first term in (\ref{both}).
Recall that
$$
\mr{\gamma} = \frac{\mathring{\rho}}{\mathring{p}}\,\frac{d \mathring{p}}{d
  \mathring{\rho}}
$$
is the adiabatic index of the reference body, cf. \eqref{eq:gammadef}.
\begin{lemma} \label{lem:radial}
Let $\xi^i = F(R) n^i$ with $F$ a positive, smooth function defined on $[0,
  R_0]$. Assume $l_i[\xi] = 0$,  Then
\begin{equation}\label{abc-new}
-(\xi, L[\xi]) =
4 \pi \int_0^{R_0} \left [  (a R^2 F'^2 + 2 b R F F' + c F^2)
+ \frac{8\mr{\rho}\,\mr{l}}{3}\left(F' - \frac{F}{R}\right)^2 \right]dR\,,
\end{equation}
where
\begin{equation}\label{abc2}
a = \mr{p} \mr{\gamma}, \quad b = 2\mr{p} (\mr{\gamma} - 2), \quad c = 4
\mr{p}(\mr{\gamma} -1) .
\end{equation}
\end{lemma}

\begin{remark}\label{rem:stability}
The elastic term in \eqref{abc-new} is clearly non-negative, and zero only
for a dilation $\xi^i \partial_i = R\partial_R$. The integrand
originating from the fluid and gravitational terms
is a quadratic form in $(RF',F)$ with determinant
$ac-b^2 = 4 \mathring{p}^2(3 \gamma -
4)$, and trace $a+c = \mr{p} [ 2\mr{\gamma} + (3\mr{\gamma} - 4)]$.
Thus, with $\mr{p} > 0$, the condition
\begin{equation}\label{eq:gamma-stab}
3\mr{\gamma} - 4 > 0
\end{equation}
is a sufficient, but by no means necessary, condition
for \eqref{abc-new} to be positive\footnote{The role of $\mr{\gamma}$ for the stability of perfect fluids
has been known since the late 19th century, see \cite{stellar}.}.
\end{remark}

\begin{proof}
Integrate $L_i[\xi]$ over $\mathcal{B}$ against $\xi^i$, using the form (\ref{lambda5}). The result, using $l_i[\xi] = 0$, is
\begin{multline}\label{lambda7}
 - (\xi^i, L \xi) = \int_\mathcal{B}\left[\left(- \mathring{p} + \mathring{\rho}\, \frac{d \mathring{p}}{d \mathring{\rho}}\right)
(\xi ^j{}_{,j})^2 +
\mathring{p}\,\xi_{i,j}\, \xi^{j,i} + 4 \mathring{\rho}
\,\mathring{l}\,(\delta_\xi \mr{\sigma})_{ij}(\delta_\xi \mr{\sigma})^{ij} \right]dX - \\
- G \int_{\mathcal{B}\times\mathcal{B}}\!\!\!\!\! \mathring{\rho}(X)\mathring{\rho}(X')
\left(\partial_i \partial_j \frac{1}{|X-X'|}\right) \xi^i(X)\,[\xi^j(X) - \xi^j(X')]\,dX dX' .
\end{multline}
In order to carry out the integration with respect to $X'$
in the second term in the last line of (\ref{lambda7}) with
$\xi^i (X) = F(R)\, n^i$,
we make use of the identity
\begin{equation}\label{idrad}
- \,\partial_i \partial_j \int_\mathcal{B} \frac{\mathring{\rho}(R')\,\xi^j(X')}{|X - X'|} \,dX' = 4 \pi \mathring{\rho}(R) \,\xi_i(X) .
\end{equation}
To see that this holds, we note that the divergence of the left side is equal to that of the right side, and make use of the fact that there
is no spherically symmetric, divergence-free, non-vanishing  vector field on $\mathcal{B}$
which is regular at the origin.
We also have that
\begin{equation}\label{m}
\partial_i \int_\mathcal{B} \frac{\mathring{\rho}(R')}{|X - X'|}\, dX' = -
\frac{n_i}{R^2} \,\mathring{\Mass}(R)\,,
\end{equation}
where $\mathring{\Mass}(R) = 4 \pi \int_0^R R'^2 \mathring{\rho}(R')\, dR'$, and
\begin{equation}\label{contribution}
\partial_i \partial_j \int_\mathcal{B} \frac{\mathring{\rho}(R')}{|X - X'|}\, dX' = - \frac{\delta_{ij} - 3 n_i n_j}{R^3}\,\mathring{\Mass}(R)
- 4 \pi \mathring{\rho}\, n_i n_j ,
\end{equation}
so that
\begin{multline}\label{contribution1}
G \,\mathring{\rho}(R) \int_{\mathcal{B}}\left(\partial_i \partial_j \frac{\mathring{\rho}(R')}{|X-X'|}\right)
[\xi^j(X) - \xi^j(X')]\,dX' =\\
= - \,G \,\mathring{\rho}(R)\, \xi^j \,\frac{\delta_{ij} - 3 n_i n_j}{R^3}\, \mathring{\Mass}(R) = 2 G \,\frac{\mathring{\rho}(R)
 \mathring{\Mass}(R)}{R^3}\,\xi_i .
\end{multline}
Here we have used that the contribution of the last term in (\ref{contribution}) cancels that from (\ref{idrad}). Evaluating the first line in
(\ref{lambda7}) is of course completely straightforward. Thus for the non-elastic part of (\ref{lambda7}) we find
the expression
\begin{equation}\label{Brad}
4 \pi \!\!\int_0^{R_0}\!\!\!\! R^2 \left[\left(\mathring{\rho}\,\frac{d \mathring{p}}{d \mathring{\rho}} - \mathring{p}\right)
\left(F' + \frac{2 F}{R}\right)^2 +
\mathring{p}
\left(F'^2 + \frac{2 F^2}{R^2}\right)
- \frac{2 G \mathring{\rho}\,\mathring{\Mass} F^2}{R^3}\right]\!\! dR  .
\end{equation}
But we know that
\begin{equation}\label{back1}
\mathring{p}' = - \frac{G \,\mathring{\rho}\,\mathring{\Mass}}{R^2} \, ,
\end{equation}
cf. \eqref{eq:dpdr}.
Inserting (\ref{back1}) into (\ref{Brad}) and integrating the last term by
parts using $\mathring{p}\,(R_0) = 0$
gives the result.
\end{proof}
Motivated by lemma \ref{lem:radial} and the discussion in remark \ref{rem:stability}, we make the following definition. 
\begin{definition}\label{def:rad-stab}
The reference body satisfies the \emph{radial stability} condition if \eqref{eq:gamma-stab} holds, or more generally, if the quadratic form in \eqref{abc-new} is positive definite.
\end{definition}
We now turn to the non-radial modes. These are analyzed using the form (\ref{id7}) and (\ref{rewrite1}) for
$L_i[\xi]$ respectively $l_i[\xi]$, for a non-radial vector field $\xi^i$,
under the assumption $l_i[\xi] =0$.

Define
\begin{equation} \label{new}
\delta_\xi \mr{\mu} = \left(\frac{d \mr{p}}{d \mr{\rho}}\,
\delta_\xi \mr{\rho} + 4 \mr{\rho}\, \mr{l} \,(\delta_\xi \mr{\sigma})_{ij}\, n^i n^j \right)\!\!\bigg{|}_{\partial \mathcal{B}}\,,
\end{equation}
Then, cf. \eqref{lambda5'}, the condition $l_i[\xi] = 0$, implies
\begin{equation}\label{also}
l_i[\xi] n^i = \mr{\rho}\, \mr{U}' \xi^i n_i|_{\partial \mathcal{B}} + \delta_\xi \mr{\mu} = 0 .
\end{equation}
Further, let $\mathfrak{U}$ and $\mathfrak{u}$ denote the Newtonian volume
and single layer potentials, respectively, i.e.
\begin{align}\label{bulk}
\mathfrak{U}[f] &=  -\, G\int_{\mathcal{B}}\!\frac{f(X')}{|X-X'|}\,dX' ,
\\
\intertext{and}
\label{nonbulk}
\mathfrak{u}[f] &=  - G\int_{\partial \mathcal{B}}\!\frac{f(X')}{|X-X'|}\,dS(X') .
\end{align}
For convenience, we have not included the factor $1/4\pi$ here.
Then we can write (see proposition \ref{prop:alternative})
\begin{equation}\label{mathfrak}
\delta_\xi \mr{U} = \mathfrak{U}\,[\delta_\xi \mr{\rho}] - \mathfrak{u}[\mr{\rho}\,\xi^i n_i] =
\mathfrak{U}\,[\delta_\xi \mr{\rho}] + \mathfrak{u}\left[\frac{\delta_\xi \mr{\mu}}{\mr{U}'}\right] -
\mathfrak{u}\left[\frac{l_i[\xi] n^i}{\mr{U}'}\right] .
\end{equation}
The following proposition, which is proved starting from \ref{id7} by a straightforward partial integration and making use of the above definitions,
provides a formula for an expression which is essentially the second variation
at the reference deformation $\phi = \id$ of the material form of the action given in equation \eqref{eq:action-material}.
\begin{prop}
Let $\xi^i$ be a vector field such that $l_i[\xi] = 0$.
Then we have
\begin{multline}\label{general}
- (\xi, L[\xi]) = \int_{\mathcal{B}}\frac{1}{\mr{\rho}}\frac{d \mr{p}}{d\mr{\rho}}
(\delta_\xi \mr{\rho})(\delta_\xi\mr{\rho}) dX + \int_{\partial \mathcal{B}} \frac{1}{\mr{\rho}\,\mr{U}'}(\delta_\xi \mr{\mu})(\delta_\xi \mr{\mu}) dS(X) + \\
+ \int_{\mathcal{B}} (\delta_\xi \mr{\rho})\left(\mathfrak{U}\,[\delta_\xi \mr{\rho}] + \mathfrak{u}\left[\frac{\delta_\xi \mr{\mu}}{\mr{U}'}\right]\right) d X +
\int_{\mathcal{\partial B}}\frac{1}{\mr{U}'} \,(\delta_\xi \mr{\mu})\left(\mathfrak{U}\,[\delta_\xi \mr{\rho}] + \mathfrak{u}\left[\frac{\delta_\xi \mr{\mu}}{\mr{U}'}\right]\right) dS(X) + \\
+ 4 \int_{\mathcal{B}}\!\!\mr{\rho}\,\mr{l}\, (\delta_\xi \mr{\sigma})_{ij}(\delta_\xi \mr{\sigma})^{ij} dX .
\end{multline}
\end{prop}
\begin{remark}
In addition to the elastic ``bulk'' contribution (third term
in the first line in (\ref{general})), there are elastic contributions to all boundary terms in (\ref{general}) resp. (\ref{nonbulk}) via $\delta_\xi \mr{\mu}$.
In the absence of elasticity and when $\mr{\rho}(R_0)=0$, after
insertion of (\ref{also}), the expression (\ref{general}) boils down to the
famous ``Chandrasekhar energy'', cf. \cite[(5-49)]{B&T:1987gady.book.....B}.
Specializing to the radial case with $l_i[\xi] = 0$, the expression \eqref{general} of course
coincides with \eqref{abc-new}. For this case it turns out to be simpler to derive
\eqref{abc-new} directly.
\end{remark}
The following result is a generalization
of the Antonov-Lebovitz theorem, cf. \cite[\S 5.2]{B&T:1987gady.book.....B} which allows for the
the case where $\mr{\rho}(R_0) > 0$.
\begin{prop} \label{prop:nonradial}
Let $\xi^i$ be a non-radial vector field such that $l_i[\xi] = 0$.
Then the quadratic form $- (\xi, L \xi)$ is non-negative, and zero
if and only if $\xi^i$ is a Euclidean Killing vector field.
\end{prop}
\begin{remark}
Our proof of proposition \ref{prop:nonradial}
follows closely the proof
valid
under the assumption $\mr{\rho}(R_0)=0$ given in the book of Binney and Tremaine, cf. \cite[Appendix 5.C]{B&T:1987gady.book.....B}.
The new edition, cf.  \cite[Appendix H]{Binney:new}
proves this result using a more direct argument due to \cite{aly}. We have not been able to generalize that argument to the case $\mr{\rho}(R_0) > 0$.
\end{remark}

\begin{proof}
Note first of all that the last (purely elastic) term in (\ref{general}) is non-negative and vanishes only on
infinitesimal conformal motions.
The remaining terms, on the other hand, depend only on pairs of functions, namely $(\delta_\xi \mr{\rho}, \delta_\xi \mr{\mu})$, on $\mathcal{B} \times \partial \mathcal{B}$. Our claim amounts to the fact that, when restricted to pairs $(\delta_\xi \mr{\rho}, \delta_\xi \mr{\mu})$ with zero spherical mean,
the sum of these terms is non-negative and zero iff $(\delta_\xi \mr{\rho}, \delta_\xi \mr{\mu}) =
(\mr{\rho}' (c,n), - \mr{\rho}\, \mr{U}' (c,n)|_{\partial \mathcal{B}})$ for some constant vector $c$.
To see that we take pairs $f = (\delta_\xi \mr{\rho},\delta_\xi \mr{\mu})$ to lie in
the Hilbert space $\mathbb{H}$ defined by
\begin{equation}\label{Hilbert}
\mathbb{H} = L^2\left(\mathcal{B}, \frac{1}{\mr{\rho}}\frac{d \mr{p}}{d \mr{\rho}}\,dX\right) \oplus L^2
\left(\partial \mathcal{B}, \frac{1}{\mr{\rho} \,\mr{U}'} \,dS(X)\right) ,
\end{equation}
with scalar product
\begin{equation}\label{scalar}
\langle f_1|f_2 \rangle =
\int_{\mathcal{B}} \frac{1}{\mr{\rho}}\frac{d \mr{p}}{d \mr{\rho}}\,\,(\delta_1 \rho)\, (\delta_2 \rho) \,dX
+ \int_{\partial \mathcal{B}} \frac{1}{\mr{\rho} \,\mr{U}'}\,(\delta_1 \mu)\,
(\delta_2 \mu) \,\,dS(X) ,
\end{equation}
and consider the operator $\mathbb{V}: \mathbb{H} \rightarrow \mathbb{H}$,
defined by
\begin{equation}
\mathbb{V} (\delta \mr{\rho},\delta \mr{\mu}) =
\left(\mathring{\rho} \frac{d \mr{\rho}}{d \mr{p}}\left(\mathfrak{U}\,[\delta_\xi \mr{\rho}] + \mathfrak{u}\left[\frac{\delta_\xi \mr{\mu}}{\mr{U}'}\right]\right), \mathring{\rho}\left(\mathfrak{U}\,[\delta_\xi \mr{\rho}] + \mathfrak{u}\left[\frac{\delta_\xi \mr{\mu}}{\mr{U}'}\right]|_{\partial \mathcal{B}}\right)\right) .
\end{equation}
The operator $\mathbb{V}$ is self-adjoint with respect to the scalar product
in (\ref{Hilbert}), due to the symmetry of the Poisson kernel.

By Lemma \ref{lem:VV}, the
volume potential $\mathfrak{U}$ defines a continuous map $L^2(\Bo) \to
H^2(\Bo)$. Further, the layer potential $\mathfrak{u}$ defines a continuous map
$L^2(\partial\Bo) \to H^1(\partial \Bo)$, cf. Lemma \ref{lem:SS}.
Since the scalar product
\eqref{scalar} defines a norm on $\mathbb{H}$ which is equivalent
equivalent to the standard norm on $L^2(\Bo) \times L^2(\partial\Bo)$, the
operator $\mathbb{V} : \mathbb{H} \to \mathbb{H}$ is compact.
Furthermore it maps $(\delta \mr{\rho},\delta \mr{\mu})$'s with vanishing spherical mean into themselves and its corresponding restriction
is also self-adjoint and compact. Now the expression in (\ref{general}) minus the elastic term take the form
\begin{equation}\label{form}
\langle f|(\mathbb{E} +  \mathbb{V}) f \rangle .
\end{equation}
Since the operator $\mathbb{V}$ is compact and self-adjoint, it has a
complete set
of eigenfunctions with real eigenvalues.
We must show that for $-\lambda$ in the spectrum of $\mathbb{V}$
it holds that $\lambda \leq 1$. Thus consider
\begin{equation}\label{bulk1}
\lambda\, \delta_\xi \mr{\rho} + \mr{\rho} \frac{d \mr{\rho}}{d \mr{p}}\, \delta_\xi \mr{U} = 0, \quad \text{in }\mathcal{B} ,
\end{equation}
and
\begin{equation}\label{bound}
\lambda\,\delta_\xi \mr{\mu} + (\mr{\rho}\,\delta_\xi \mr{U})|_{\partial \mathcal{B}}=0, \quad \text{in }\partial \mathcal{B} ,
\end{equation}
where $\delta_\xi U = \mathfrak{U}\,[\delta_\xi \mr{\rho}] + \mathfrak{u}\left[\frac{\delta_\xi \mr{\mu}}{\mr{U}'}\right]$ . We write (\ref{bulk1}) in the form
\begin{equation}\label{rebulk}
\Delta \delta_\xi \mr{U} + \frac{4 \pi G}{\lambda}\,\mr{\rho} \frac{d \mr{\rho}}{d \mr{p}}\, \delta_\xi \mr{U} = 0 .
\end{equation}
Furthermore we write $\delta_\xi \mr{U}$ in the form
\begin{equation}\label{delta U}
\delta_\xi \mr{U} = \mr{U}' \delta s = \mr{U}' \sum_l s_l(R)\, Y_l (\Omega)\,,
\end{equation}
where we suppress the index $m$ of spherical harmonics.

Using the radial derivative of $\mr{U}'' + \frac{2}{R} \mr{U}' = 4 \pi G \mr{\rho}$ to eliminate $\mr{U}'''$, it is straightforward to show that (\ref{rebulk}) can be written in the form
\begin{equation}\label{eigenl}
\frac{1}{R^2 \mr{U}'}(R^2 \mr{U}'^2 s_l')' - \frac{l^2 + l - 2}{R^2}\, \mr{U}' s_l - 4 \pi G\,\mr{\rho} \frac{d \mr{\rho}}{d \mr{p}}\,\mr{U}'
\left(1- \frac{1}{\lambda}\right) s_l = 0 .
\end{equation}
Integrating (\ref{eigenl}) over $(0,R_0)$ against $R^2 \mr{U}' s_l$ we find that
\begin{multline}\label{eigenint}
\int_0^{r_0} R^2 \mr{U}'^2 \left[ s_l'^2 + \left(\frac{l^2 + l - 2}{r^2} + 4 \pi G \,\mr{\rho} \frac{d \mr{\rho}}{d \mr{p}}\,
\left(1 - \frac{1}{\lambda}\right)\right)s_l^2\right] dR \\
- R_0^2 (\mr{U}'^2 s_l s_l')(R_0) = 0 .
\end{multline}
Next observe that the expression for $\delta_\xi \mr{U}$ makes sense both in the interior and exterior region, and that
$\delta_\xi \mr{U}$ is continuous across $\partial \mathcal{B}$ and $\delta_\xi \mr{U}'$ suffers a jump
\begin{equation}
[(\delta_\xi \mr{U})']_{\partial B(r_0)} = 4 \pi G \frac{\delta_\xi \mr{\mu}}{\mr{U}'(R_0)} = - \frac{4 \pi G}{\lambda} \,
\frac{\mr{\rho} \,\delta_\xi \mr{U}}{\mr{U}'}\bigg{|}_{\partial \mathcal{B}} =
- \frac{4 \pi G}{\lambda}\,\mr{\rho}\, \delta s|_{\mathcal{B}} ,
\end{equation}
where we have used (\ref{bound}) in the second equality. Since
\begin{equation}
[\mr{U}'']|_{\partial \mathcal{B}} = - 4 \pi G \mr{\rho}|_{\partial \mathcal{B}} ,
\end{equation}
it follows that
\begin{equation}
[\delta s']_{\partial \mathcal{B}}= 4 \pi G \left(1 - \frac{1}{\lambda}\right)\left(\frac{\mr{\rho} \,\delta s}{\mr{U}'}\right)\Big{|}_{\partial \mathcal{B}}
\end{equation}
This in turn implies that
\begin{equation}
[s_l'](R_0)= 4 \pi G \left(1 - \frac{1}{\lambda}\right)\left(\frac{\mr{\rho} \,s_l}{\mr{U}'}\right)(R_0) .
\end{equation}
But, by virtue of the multipole expansion in the exterior region, we also have that
\begin{equation}
\lim_{R \downarrow R_0}\left((\mr{U}'s_l)' + \frac{l+1}{R_0}\,\mr{U}' s_l\right)=0\,\Rightarrow
\lim_{R \downarrow R_0}\left(s_l' + \frac{l-1}{R_0}\,s_l\right)=0\,.
\end{equation}
Thus
\begin{equation}\label{insert}
\lim_{R \uparrow R_0}\left(s_l' + \frac{l-1}{R_0}\,s_l\right)=4 \pi G \left(\frac{1}{\lambda} - 1\right)\left(\frac{\mr{\rho} \,s_l}{\mr{U}'}\right)(R_0) .
\end{equation}
Inserting (\ref{insert}) into (\ref{eigenint}), there results
\begin{multline}
\int_0^{R_0}\!\!\!\!R^2 \mr{U}'^2 \left[ s_l'^2 + \frac{l^2 + l - 2}{R^2}\,s_l^2\right]dR +
R_0 (l-1)(\mr{U}'^2\,s_l^2)(R_0) =\\
= 4 \pi G \left(\frac{1}{\lambda} - 1\right)\left[\int_0^{R_0}\!\!\!\mr{\rho} \frac{d \mr{\rho}}{d \mr{p}}\,
s_l^2\, dR + R_0^2 (\mr{U}'\mr{\rho} \,s_l^2)(R_0)\right] .
\end{multline}
Thus (note that $\mr{U}'>0$ and, since $\delta s$ has zero spherical mean,
$l \geq 1$) it follows that $\lambda \leq 1$ and $\lambda =1$ implies that
$s_l = 0$ for $l>1$ and $s_1 = \mathrm{const}$. This in turn means that $\delta_\xi \mr{U} = \mr{U}'(c,n) = c^i \partial_i \mr{U}$ which implies
\begin{equation}
\Delta \delta_\xi \mr{U} = 4 \pi G \,\mr{\rho}'(c,n)\,.
\end{equation}
Thus $\delta_\xi \mr{\rho} = \mr{\rho}'(c,n)$. It now follows that $\mathbb{V}$ restricted to quantities with zero spherical mean has eigenvalues $\lambda \leq
1$ and $\lambda = 1$ implies that $\delta_\xi \mr{\rho} = c^i \mr{\rho}_{,i}$. Furthermore (\ref{bound}) now implies
$\delta_\xi \mr{\mu} = - \mr{\rho} \,\mr{U}' (c,n)|_{\partial \mathcal{B}}$.\\
So far $\xi$ itself was not involved in the argument. Going back to the definition (\ref{new}) we infer that
$(\delta_\xi \mr{\sigma}_{ij} n^i n^j|_{\partial \mathcal{B}} = 0$.
Now, due to the presence of the elastic term in (\ref{general}), it follows that $(\delta_\xi \mr{\sigma})_{ij} = 0$
everywhere in $\mathcal{B}$, whence $\xi$ is a conformal Killing vector. Dilations and conformal boosts are incompatible with $\delta_\xi \mr{\rho} = c^i \mr{\rho}_{,i}$.
The only remaining possibility is that $\xi^i = c^i + \Omega^i{}_j X^j$ with $\Omega_{ij} = \Omega_{[ij]}$.
\end{proof}

Summing up we obtain the
\begin{thm}\label{thm:kernel}
Assume that the radial stability condition, cf. definition \ref{def:rad-stab}, holds.
Then the nullspace of $L_i$,
under the condition that $l_i$ be zero consists exactly of infinitesimal
Euclidean motions.
\end{thm}
As a final remark note that the Antonov-Lebovitz (i.e. non-radial-non-radial) part of the previous argument does not require any
condition on the background equation of state involving $\gamma$. It implies in particular that the kernel of the pure fluid
linearized operator acting on non-radial modes is trivial, i.e. only consists of $\xi$'s being translations or satisfying $\delta_\xi \mr{\rho} = 0$.
This in turn implies that there are no nontrivial  nonspherical perturbations of a Newtonian perfect-fluid star
with given equation of state - which in turn is a linearized version of the classical spherical-symmetry result of Lichtenstein.

\section{Fredholm alternative}
In this section we consider the operator
\newcommand{\Lbb}{\mathbb L}
\newcommand{\Lbbloc}{\mathbb L^{loc}}
\newcommand{\Lloc}{L^{loc}}
\newcommand{\Abb}{\mathbb A}
\newcommand{\Zbb}{\mathbb Z}
\newcommand{\Zcal}{\mathcal Z}
$$
\Lbb: \xi \to (-L[\xi], l[\xi]) .
$$
We make use of the results and methods of \cite[chapitre 2]{lions:magenes:vol1}.

Let $\Lloc$ denote the local part of the operator $L$, corresponding to the
first line of \eqref{lambda5}. Then
$\xi \to \Lbbloc \xi = (-\Lloc[\xi] , l[\xi])$ is
an elliptic, formally self-adjoint differential operator of second
order. $\Lbbloc$ is therefore Fredholm as a map
$$
H^{2+k}(\Bo) \to H^k(\Bo) \times H^{1/2+k}(\partial \Bo) ,
$$
for $k \geq 0$, $k$ integer.

Further, let $Z$ be the gravitational part of $-L$, given by
the second line in \eqref{lambda5}.
Explicitely, for $X \in \Bo$,
\begin{align*}
Z_i [\xi](X) &= G \xi^j(X) \partial_i \partial_j \mathring U(X) -
- G \rhonull(X) \int_\Bo \left ( \partial_i \partial_j \frac{1}{|X - X'|}
\right )
\rhonull(X') \xi^j(X') dX \\
&=  G \xi^j(X) \partial_i \partial_j \mathring U
+ 4\pi G \rhonull(X) \partial_i \partial_j
\left ( \Delta^{-1} (\rhonull \xi^j \chi_{\Bo}) \right ) (X) .
\end{align*}
By assumption, the reference density $\rhonull$ is smooth on $\overline
\Bo$. By Lemma \ref{lem:VV}, $Z$ is a bounded operator
$$
Z: H^k(\Bo) \to H^k(\Bo) ,
$$
for $k \geq 0$, $k$ integer.
Hence, $\xi \mapsto (Z[\xi],0)$ is a compact linear map
$$
H^{2+k}(\Bo) \to H^k(\Bo) \times H^{1/2+k}(\partial \Bo) ,
$$
and hence it follows that $\Lbb$ is  Fredholm, since a compact perturbation of
a Fredholm operator is again Fredholm,
cf. \cite[Theorem 37.5]{conway:operatortheory:MR1721402}.
In particular, $\Lbb$ has
finite dimensional kernel and closed range with finite dimensional cokernel
given by $\ker \Lbb^*$ where $\Lbb^*$ is the operator mapping
$$
H^{-k}(\Bo) \times H^{-1/2-k}(\partial\Bo) \to H^{-2-k}(\Bo) ,
$$
defined by
$$
\la \Lbb \xi , \Phi \ra = \la \xi , \Lbb^* \Phi \ra ,
$$
where $\la \cdot , \cdot \ra$ is the duality pairing.
Explicitely, with $\Phi = (v,\varphi)$ we have
$$
\la \Lbb \xi , \Phi \ra = \int_\Bo - L_i[\xi] b^i dx + \int_{\partial\Bo}
l_i[\xi] \tau^i dS ,
$$
\newcommand{\range}{\text{range}\,}
and $\range\Lbb$ is given by the space of $F = (b, \tau)$ such that
$$
\int_\Bo b_i v^i dx + \int_{\partial \Bo} \tau_i \varphi^i dS  = 0 ,
$$
for all $\Phi \in \ker \Lbb^*$.
We have the Green's identity
$$
\int_\Bo \eta^i L_i[\xi] d x - \int_\Bo L_i[\eta] \xi^i d x
= \int_{\partial\Bo} \eta^i l_i[\xi] dS(x) - \int_{\partial\Bo} l_i[\eta]
\xi^i dS .
$$
Now consider the case $k=0$.
Since $\Phi \in \ker \Lbb^*$ is equivalent to
$$
\la \xi , \Lbb^* \Phi \ra = \la \Lbb \xi , \Phi \ra = 0 ,
$$
for all $\xi \in H^2(\Bo)$, we have for $\Phi = (v, \varphi) \in \ker
\Lbb^*$,
\begin{align*}
0 &= \int_\Bo -L_i[\xi] v^i dx + \int_{\partial \Bo} l_i[\xi] \varphi^i dS \\
&= \int_\Bo - \xi^i L_i[v] dx + \int_{\partial\Bo} \xi^i l_i[v] + \int_{\partial \Bo} (\varphi^i - v^i)
l_i[\xi] dS .
\end{align*}
Since $\xi \in H^2(\Bo)$ is arbitrary, this gives immediately that $L_i[v]=
0$. Further, we have that $(\xi, l[\xi])$ is a Dirichlet system on $\partial \Bo$, and hence it
follows that $l_i[v] = 0$ and $\varphi^i = v^i
$ on $\partial\Bo$.
It follows from the analysis in \cite{lions:magenes:vol1} that the kernel and
cokernel of $\Lbb$ are independent of $k$, and hence are well-defined. We end
up with the following result.
\begin{prop} \label{prop:fredholm}
Assume that the radial stability condition, cf. definition \ref{def:rad-stab}, holds. 
Let $k \geq 0$, $k$ integer.
The operator  $\xi \mapsto \Lbb[\xi] = (-L[\xi], l[\xi])$
is a Fredholm operator
$$
H^{2+k}(\Bo) \to H^k(\Bo) \times H^{1/2+k}(\partial \Bo) ,
$$
with finite dimensional kernel, and range consisting of $(b_i, \tau_i)$ such
that
$$
\int_{\Bo} -b_i v^i dx + \int_{\partial\Bo} \tau_i v^i dS = 0 ,
$$
for all $v \in \ker L$.
\end{prop}

\section{Implicit function theorem} \label{sec:implicit}
In this section, we assume that the radial stability condition, cf. definition \ref{def:rad-stab}, holds.
Let $\rhonull_\lambda$, $\lambda \in (-\eps,\eps)$
be a 1-parameter family of densities on $\Bo$.
Let $\FF[\rhonull_\lambda,\phi] = (b_i[\rhonull_\lambda,\phi],
\tau_i[\rhonull_\lambda,\phi])$ where
\begin{align*}
b_i[\rhonull,\phi] &= - \partial_A (\rhonull \tau^A{}_i) +
\rhonull (\partial_i U) \circ \phi  , \\
\tau_i [\rhonull,\phi] &= \tau^A{}_i n_A \big{|}_{\partial\Bo} .
\end{align*}
The Frechet derivative of $\FF$ at the reference state $\lambda=0, \phi=\id$
is
$$
D_\phi \FF\big{|}_{\lambda = 0, \phi = \id} \, \xi  = \Lbb \, \xi = (-L[\xi], l[\xi]) .
$$
By theorem \ref{thm:kernel}, and proposition \ref{prop:fredholm} we have that
the kernel and cokernel
of $D_\Phi \FF$ consists of Killing vector fields, i.e. $\zeta^i$ of the form
$c^i + \omega^i{}_j X^j$ for
constant $c^i, \omega^i{}_j$, $\omega_{ij} = \omega_{[ij]}$.

\newcommand{\Pbb}{\mathbb P}
Now let $\Pbb$ be the projection defined by
$\PP: (b_i, \tau^i)  \mapsto ({b'}_i, \tau_i)$, with
$$
{b'}_i = b_i + c_i + \omega_{ij} X^j ,
$$
for suitable $c_i, \omega_{ij}$ such that $({b'}^i, \tau^i)$
 are equilibrated, i.e.
$$
\int_{\Bo} {b'}_i \zeta^i dx - \int_{\partial\Bo} \tau_i \zeta^i dS = 0 ,
$$
for all $\zeta^i \in \ker \Lbb$.
The above determines $b'_i$ in terms of $b_i$ and $\tau_i$ and hence also the projection
operator
$\Pbb$.

\newcommand{\Xbb}{\mathbb X}
\newcommand{\Ybb}{\mathbb Y}
We eliminate the kernel of $\Lbb =
D_\phi \FF \big{|}_{\lambda = 0, \phi=\id}$ by
fixing the 1-jet of $\phi$ at the origin. Denote the space of $\xi^i \in
H^{2+k}(\Bo)$ with $\xi_i(0) = 0$, $\partial_i \xi_j(0) = 0$ by $\Xbb$ and let
$\Ybb$ denote the range of $\Lbb$. Then we have that $\Lbb: \Xbb \to \Ybb$ is
an isomorphism. Now the following result is an immediate consequence of the
implicit function theorem.
\begin{prop} \label{prop:implicit}
Fix $k \geq 1$ and $\eps > 0$.
For $\lambda \in (-\eps, \eps)$,
let  $\lambda \mapsto \rhonull_\lambda$ be a 1-parameter
family of smooth functions on $\Bo$, so that $\rhonull = \rhonull_0$ is a
solution to $\FF[\rhonull, \id] = 0$.
There is an $\eps_0 > 0$ so
that  for $\lambda \in (-\eps_0, \eps_0)$, there is a
unique $\phi_\lambda \in H^{2+k}(\Bo)$ with $\phi(0) = 0$, $\partial_j
\phi^i(0) = \delta^i{}_j$, such that
$$
\Pbb \FF[\rhonull_\lambda, \phi_\lambda] = 0 .
$$
\end{prop}

\subsection{Equilibration}
Recall that the the field equation in the Eulerian picture takes the form
$$
\nabla_i (\tau^i{}_j \chi_{f^{-1}(\Bo)} + \Theta^i{}_j) = 0 .
$$
For $\zeta^i$ a Killing field we have $\nabla_{(i} \zeta_{j)} = 0$ and hence
$$
\int_{\Re^3} \zeta_j \nabla_i (\tau^{ij} \chi_{f^{-1}(\Bo)} +
\Theta^{ij}) dx  = - \int_{\Re^3} \nabla_i \zeta_j (\tau^{ij}
\chi_{f^{-1}(\Bo)} + \Theta^{ij}) dx = 0 ,
$$
since $\tau^{ij} = \tau^{(ij)}$ and $\Theta^{ij} = \Theta^{(ij)}$.
Recall
$$
\nabla_i \Theta^i{}_j = \rho \chi_{f^{-1}(\Bo)} \partial_j U .
$$
Applying the change of variables formula and the Piola identity, and taking
into account the boundary condition
$\tau^i{}_j n_i \big{|}_{\partial f^{-1}(\Bo)} = 0$ we have
$$
0 = \int_{\Bo} \zeta^j \circ \phi  [- \partial_A \taudens^A{}_j + (\partial_j
U)\circ \phi ] dX ,
$$
and hence $(b_i, 0)$ is equilibrated with respect to $\zeta^i \circ \phi$.

We are now ready to prove
\begin{thm} \label{thm:mainimplicit}
Assume that the radial stability condition holds.
Fix $k \geq 1$ and $\eps > 0$.
For $\lambda \in (-\eps, \eps)$,
let  $\lambda \mapsto \rhonull_\lambda$ be a 1-parameter
family of smooth functions on $\Bo$, so that $\rhonull = \rhonull_0$ is a
solution to $\FF[\rhonull, \id] = 0$, and $\rhonull_\lambda \big{|}_{\partial
  \Bo} = \rho_0$ for $\lambda \in (-\eps, \eps)$. There is an $\eps_0 > 0$ so
  that  for $\lambda \in (-\eps_0, \eps_0)$, there is a
unique $\phi_\lambda \in H^{2+k}(\Bo)$ with $\phi(0) = 0$, $\partial_j
\phi^i(0) = \delta^i{}_j$, such that
$$
\FF[\rhonull_\lambda, \phi_\lambda] = 0 .
$$
\end{thm}
\begin{proof} Let $\phi_\lambda$ be the solution to $\Pbb
  \FF[\rhonull_\lambda, \phi_\lambda] = 0$ constructed in
proposition \ref{prop:implicit}, and let $K$ denote the space of Killing fields on $\Re^3$.
By the proof of
proposition \ref{prop:implicit} we have that the load
$b^i = b^i[\rhonull_\lambda, \phi_\lambda]$ corresponding to
$\phi_\lambda$ satisfies $b^i \in K$. On the other hand, we have by the above
discussion that
$$
\int_\Bo \zeta^i \circ \phi_\lambda b_i = 0 , \quad \forall \zeta^i \in K.
$$
For $\phi_\lambda$ sufficiently close to $\id$ this implies $b_i = 0$.
Since $\phi_\lambda$ depends continuously on $\lambda$, the result follows.
\end{proof}

\section{Non-spherical nature of solutions} \label{sec:physint}
It is important to understand that the method we have developed in this work is capable  of ``building mountains'', i.e. that the solutions we construct are indeed non-spherical. Before proving that
this is the case, it will be useful to make a few observations on the pure fluid case, where the action in the material frame \eqref{eq:action-material} takes the (static Euler-Poisson) form
\begin{multline}\label{euler-poisson}
\mathcal{L}_{ep}[\mathring{\rho};\phi]= \int_{\mathcal{B}}\mathring{\rho}(X)
\,e(\mathring{\rho} \, \det( \phi^j{}_B)^{-1})(X)
dX \\
- \frac{G}{2}\int_{\mathcal{B}\times\mathcal{B}}\!\!\!\!\mathring{\rho}(X)
\mathring{\rho}(X')\,\frac{dX\, dX'}{|\phi(X) - \phi(X')|} .
\end{multline}
Let  $\psi$ be a diffeomorphism $\psi:\Bo \rightarrow \Bo$
(in particular $\psi$ maps $\partial \mathcal{B}$ into itself).
Further, let $\mr{\rho}_\psi = \mathring{\rho}\circ \psi\,|\frac{\partial \psi}{\partial X}|,$
where $|\frac{\partial \psi}{\partial X}| = \det(\psi^i{}_A)$,
and
$\phi_\psi = \phi \circ \psi$.
Then the action given by \eqref{euler-poisson}
satisfies the covariance
property
\begin{equation}\label{inv}
\LL_{ep}[\mathring{\rho}\,;\phi]=\LL_{ep}[\mathring{\rho}_\psi;\phi_\psi] .
\end{equation}
This covariance property is of course reflected by the fact that the Eulerian variable
$\rho[\mr{\rho},\phi](x) = (\mr{\rho}\,|\frac{\partial \phi}{\partial X}|^{-1})(\phi^{-1}(x))$ remains unchanged under the action of $\psi$ on $(\mr{\rho},\phi)$.
It follows that the Euler-Poisson system is symmetric\footnote{Not that \emph{covariance} refers to a change of both the field variable $\phi$ and the background field,
whereas \emph{symmetry} refers to change of field variable leaving the background invariant.
For example,
special relativistic Maxwell theory (the background field in this case being the Minkowski metric)  is covariant under spacetime diffeomorphisms but symmetric under Poincare transformations.} under the action of the infinite dimensional group of volume preserving diffeomorphisms $\psi$
of $\mathcal{B}$, i.e. when $\psi$ leaves the volume form $\mr{\rho}\,dX$ invariant, then
$$
\mathcal{L}_{ep}[\mathring{\rho};\phi \circ \psi]= \mathcal{L}_{ep}[\mathring{\rho};\phi] .
$$
Next recall from a theorem of Moser \cite{moser} that for any positive densities $\mr{\rho}_1$, $\mr{\rho}_2$ in $\bar{\mathcal{B}}$ such that
$\int_\mathcal{B} \mr{\rho}_1 dX = \int_\mathcal{B} \mr{\rho}_2 dX$ there exists a diffeomorphism $\psi:\Bo \rightarrow \Bo$, such that $\mr{\rho}_2 = \mathring{\rho}_1\circ \psi\,|\frac{\partial \psi}{\partial X}|$. Thus, for fluids, any change of density $\mr{\rho}$ leaving its integral over $\mathcal{B}$ unchanged, should be viewed as a mere change of gauge. Note that $\int_\mathcal{B}\mr{\rho}\, dX$ is nothing but the total mass of the physical solution,
namely $\int_{\phi(\mathcal{B})} \rho[\mr{\rho},\phi](x) \,dx$.
\begin{remark}
The above transformation properties
explain the fact that the linearized operator \eqref{eq:lambda-eqs}, in the absence of the elastic terms, has a kernel containing translations plus the infinite dimensional set of vector fields
$\xi$ satisfying $\delta_\xi \mr{\rho} =0$ and $\xi^i n_i|_{\mathcal{B}} = 0$. When the
radial stability condition, cf. definition \ref{def:rad-stab}, holds,
the proofs of lemma \ref{lem:radial},
and of proposition \ref{prop:nonradial} in the case where $\mr{l}$ is zero, show that the kernel in fact consists of these vector fields. Finally, we remark that the transformation properties can be used to derive the identity
stated in proposition \ref{prop:alternative}. The proof of the last statement is left to the reader.
\end{remark}
Let $\delta \mr{\rho} = \frac{d}{d \lambda}\mr{\rho}_\lambda|_{\lambda = 0}$. Before stating our proposition on lack of spherical symmetry we prove two lemmas. The first concerns the form of the perturbed Eulerian density and Cauchy stress.
\begin{lemma} There holds
\begin{equation}\label{1steuler}
\frac{d}{d \lambda} \rho[\mr{\rho}_\lambda,\phi_\lambda] \bigg{|}_{\lambda = 0}
= \delta \mr{\rho} - \delta_\xi \mr{\rho} ,
\end{equation}
and
\begin{equation}\label{2ndeuler}
\frac{d}{d \lambda} \tau_{ij}[\mr{\rho}_\lambda,\phi_\lambda] \bigg{|}_{\lambda = 0}
= \frac{d \mr{p}}{d\mr{\rho}}(\delta \mr{\rho} - \delta_\xi \mr{\rho})\delta_{ij} -
4 \mr{\rho}\,\mr{l}\, (\delta_\xi \mr{\sigma})_{ij} .
\end{equation}
\end{lemma}
\begin{proof} The proof of \eqref{1steuler} follows easily from \eqref{eq:important}. The proof of \eqref{2ndeuler} follows from \eqref{piola}
together with \eqref{piolacauchy}.
\end{proof}
The second lemma concerns the linearized problem.
\begin{lemma}\label{lem:mountains}
The linearized problem, namely
\begin{equation}\label{lineq}
D_\phi \mathcal{F}[\mathring{\rho},{\bf{id}}] . \xi = -  D_{\mr{\rho}} \,\mathcal{F}[\mathring{\rho},{\bf{id}}]  . \delta \mr{\rho}
\end{equation}
takes the explicit form 
\begin{subequations}\label{eq:final-syst}
\begin{multline}\label{finally}
\mr{\rho}\,\partial_i\left(\frac{1}{\mr{\rho}}\frac{d \mr{p}}{d\mr{\rho}}\, \delta_\xi \mr{\rho} + \delta_\xi \mr{U}\right) +
4 \,\partial^j \left[\mathring{\rho}\, \mathring{l}\,(\delta_\xi \mr{\sigma})_{ij}\right]
\\
=
\mr{\rho}\,\partial_i\left(\frac{1}{\mr{\rho}}\frac{d \mr{p}}{d\mr{\rho}}\, \delta \mr{\rho}\right)
- G \mr{\rho}\,  \partial_i \int_{\mathcal{B}} \frac{(\delta \mr{\rho)}(X')}{|X - X'|} dX' ,
\end{multline}
for the bulk and
\begin{equation}\label{given}
\left[\left(\frac{d \mathring{p}}{d \mathring{\rho}}\,\delta_\xi \mr{\rho} + \mr{\rho}\, \mr{U}' \xi^j n_j\right)n_i
+
4 \mr{\rho}\,\mr{l}\, (\delta_\xi \mr{\sigma})_{ij} n^j
\right]\,\bigg{|}_{\partial \mathcal{B}} = \left(\frac{d \mathring{p}}{d \mathring{\rho}}\,\delta \mr{\rho}\right) n_i \bigg{|}_{\partial \mathcal{B}} ,
\end{equation}
\end{subequations}
for the boundary part.
\end{lemma}
\begin{proof}
The left hand side of \eqref{lineq} is clearly given by \eqref{id7}. To deal with the right hand side, we first note that since the invariant $\tau$ is independent of $\mathring{\rho}$, the elastic contribution to the right hand side of \eqref{finally}  is zero. Now the form of the right hand side of equation \eqref{finally} follows by explicit computation
from the fluid Euler-Lagrange equation, namely
\begin{multline}\label{fluideulag}
\mathcal{E}[\mr{\rho},\phi]_i =
- \partial_A \frac{\partial \,[\mr{\rho}\, e(\mr{\rho}\,\det(\phi^j{}_B)^{-1})]}{\partial \phi^i{}_A}
\\
- G \rhonull(X)\!\! \int_{\Bo}\!\! \left (\!\!\partial_i \frac{1}{|z|} \right ) \!\bigg{|}_{z
  = \phi(X) - \phi(X')}\!\!\! \rhonull(X') d X' ,
\end{multline}
in $\mathcal{B}$.
Finally,
the form of the right hand side in the boundary condition \eqref{given} follows easily from
\begin{equation}
D_{\mr{\rho}} \,\taudens^A{}_i[\mr{\rho},{\bf{id}}]= - \frac{d \mathring{p}}{d \mathring{\rho}}\, \delta^A{}_i ,
\end{equation}
which in turn follows from \eqref{unperturbed}.
\end{proof}
We now state our result on lack of spherical symmetry.
\begin{prop}\label{prop:nonsph}
Suppose
$\mr{l} > 0$.
and that the
stability condition \eqref{eq:gamma-stab}
is satisfied.
Let a non-zero function $\delta \mr{\rho}$ on $\Bo$ be given, which has only $l \geq 2$
nonzero modes\footnote{The index $l$ should not be confused with $l(\rho,\tau,\delta)$ or $\mr{l}$.}  in its spherical harmonics expansion, and which has  $\delta \mr{\rho}|_{\partial \mathcal{B}} = 0$.

Then
the perturbed
stress of the physical body,
given by \eqref{2ndeuler},  is not spherically symmetric. Here the perturbed stress is
 calculated with respect to the unique vector field $\xi^i$ solving equation \eqref{lineq}, constructed as in section \ref{sec:implicit}.
\end{prop}
\begin{remark}
The proof of proposition \ref{prop:nonsph} applies more generally in the case when the quadratic form, defined by \eqref{abc-new} with the elastic terms set to zero, is positive definite.
\end{remark}
\begin{proof}
We first show that with the given $\delta\mr{\rho}$, the right hand side of \eqref{finally}, which we shall denote by $H_i$, is non-zero.
Assume for a contradiction, that $H_i$ is zero.

Since $\delta\mr{\rho}$ has only $l \geq 2$ modes, it holds that
$$
\int_{\mathcal{B}} \delta \mr{\rho}\, dX = 0
$$
and hence there is a vector field $\eta^i$ such that
$\delta_\eta \mr{\rho} = \delta \mr{\rho}$ and $\eta^i n_i|_{\partial \mathcal{B}} = 0$.

Inserting $\delta_\eta\mr{\rho}$ into the right hand side of equation \eqref{finally}, and taking into account that $\delta\mr{\rho} |_{\delta\Bo} = 0$, we see that $H_i$ is identical to $L_i[\eta]$ given by \eqref{id7} for the particular case when the elastic term is absent, i.e., with $\mr{l} = 0$.

It follows from $\delta\mr{\rho} {|}_{\partial\Bo} = 0$ that the vector field $\eta^i$ satisfies
$\partial_i \eta^i|_{\partial \mathcal{B}} = 0$.
Hence we have that
$\eta$ is in the null space of the operator $L_i$ given by \eqref{id7} and satisfies the boundary condition \eqref{also}, both in the case $\mr{l} = 0$. Then the proof of proposition \ref{prop:nonradial} in the case where $\mr{l}
= 0$ shows that either $\eta$ is a translation Killing vector or $\delta_\eta \mr{\rho} = 0$. If $\delta_\eta\mr{\rho}$ we have a contradition, and hence $\eta$ is a translation. However, in this case, $\delta \mr{\rho}$ has only $l=1$ components, and it follows that $\delta \mr{\rho} = 0$, which gives a contradition.  Thus we have proved that
$H_i$ is non-zero.

Further, by construction, $H_i$ has only $l \geq 2$ components. To make concrete what this means
for a (co-)vector field, recall that any covector $\kappa_i$ can be written in the form
$$
\kappa_i = a n_i + (\delta_i{}^j - n_i n^j)\partial_j b + \epsilon_i{}^{jk}n_j \partial_k c
$$
where $a, b, c$ are scalar fields on $[0,R_0] \times \mathbb{S}^2$ with
$a$ unique, and $b,c$ unique up to constants. The triple of scalars $(a,b,c)$  corresponding to covector $H_i$ has non-zero $l \geq 2$ components in its spherical harmonics expansion. Due to the
stability condition \eqref{eq:gamma-stab},
the boundary value problem given \eqref{eq:final-syst}
has a solution $\xi^i$, which is unique up to Killing vectors.

Recall that Killing vectors have only $l \geq 1$ components. It follows that we can set the $l\leq 1$ components of $\xi^i$ to zero, and get a solution which we also denote $\xi^i$ to \eqref{eq:final-syst}, such that $\xi^i$ has only $l \geq 2$ components.

Finally we show that the the perturbed Eulerian stress tensor \eqref{2ndeuler} is not spherically symmetric. To do this we note that if it were spherically symmetric, then its tracefree part would
be of the form
\begin{equation}\label{eq:tracefree}
- 4 \mr{\rho}\,\mr{l}\, (\delta_\xi \mr{\sigma})_{ij} = A (\delta_{ij} -
3 n_i n_j)
\end{equation}
for some function $A$ depending only on $R$. Since by construction $\xi^i$ is non-zero and has only $l \geq 2$ modes, equation \eqref{eq:tracefree} can have a solution only if $A=0$ and $\xi^i$ is a conformal Killing vector. But conformal Killing vectors have only $l\leq 1$ modes, which is a contradiction.  Therefore, the perturbed Eulerian stress is not spherically symmetric, which completes the proof.
\end{proof}
We finally point out that it would be interesting to have information about the spherical behavior of $\xi^i n_i|_{\partial \mathcal{B}}$, since this describes the shape of mountains to order $\lambda$. However, this would require a detailed analysis of the boundary value problem given by \eqref{eq:final-syst},
which we defer to later work.
\subsection*{Acknowledgements} This material is based in part on work supported
by the National Science
Foundation under Grant No. 0932078 000, while LA and RB were in residence at
the Mathematical Sciences Research Institute in Berkeley, California, during
the semester programme of Mathematical General Relativity, during the fall of 2013. Further, LA and RB thank ESI for hospitality and support during part of this work. BS is grateful to J\"org Frauendiener for discussions in the early stages of the present work.

\appendix

\renewcommand{\SS}{\mathcal S}
\newcommand{\DD}{\mathcal D}
\newcommand{\VV}{\mathcal V}
\newcommand{\Tr}{\mathop{{\rm Tr}}}

\section{Estimates for the Newtonian potential} \label{sec:newtest}
Here we prove some potential theoretic estimates which are used in our
analysis. We discuss here only estimates in the setting of
$L^2$ Sobolev spaces $H^s$. See \cite[Chapter 4]{taylor:PDE:I:MR2744150} for
background. The analogous results hold in the setting of Sobolev
spaces of $L^p$ type $W^{s,p}$.

Consider $\Re^n$, $n \geq 3$, with Cartesian coordinates $(x^i)$, and with the Euclidean
metric. Let $\Omega$ be a smooth, bounded domain in $\Re^n$,  with boundary
$\partial\Omega$. Let $n^i$ be the outward directed normal to $\partial
\Omega$. The trace of a function $f$ on $\partial\Omega$ is denoted
$\Tr_{\partial\Omega} f$.

Let $\omega_n$ be the area of the unit sphere in $\Re^n$, and let
$E = -1/(\omega_n |x|^{n-2})$ be the fundamental solution of the Laplace
equation. The volume potential of a function $f$ is
$$
\VV[f](x) = \int_\Omega E(x-x') f(x') dx' ,
$$
and the layer potential $\SS[f]$ is
$$
\SS[f](x) = \int_{\partial\Omega} E(x-x') f(x') dS(x') ,
$$
where $dS$ is the induced volume element on $\partial\Omega$.
We shall need the following standard result, cf.
\cite[Chapter 7, Proposition 11.2]{taylor:PDE:II:MR1395149}. (See also
\cite[Chapter 7, Proposition 11.5]{taylor:PDE:II:MR1395149}
The assumption that the complement of $\Omega$ is connected
  made in \cite[Chapter 7, Proposition 11.5]{taylor:PDE:II:MR1395149} is
  not relevant for the continuity property which we need here.)

\begin{lemma} \label{lem:SS}
$\SS$
defines a bounded operator $\SS: H^{s-1}(\partial\Omega) \to
  H^s(\partial \Omega)$.
\end{lemma}

Let $\chi_\Omega$ denote the indicator function of $\Omega$, i.e.
$$
\chi_{\Omega}(x) = \left\{ \begin{array}{ll} 1 & x\in \Omega , \\
0 & x \notin \Omega,  \end{array} \right.
$$
and let $\delta_{\partial\Omega}$ be the delta-distribution supported on
$\partial\Omega$. Let $f$ be sufficiently regular so that the trace
$\Tr_{\partial\Omega} f$ is defined. Then the following identity is valid in
the sense of distributions.
\begin{equation}\label{eq:dchi}
\partial_i (f \chi_\Omega) = (\partial_i f) \chi_{\Omega} - \Tr[f n^i]
\delta_{\partial \Omega} .
\end{equation}
To see this, let $\phi \in C^\infty_0(\Re^n)$ and let $\xi^i$ be a vector
field on $\Re^n$. Then we
have
\begin{align*}
\int_{\Re^n} \phi \xi^i \partial_i \chi_\Omega dx &=
\int_{\Re^n} \phi \partial_i (\xi^i \chi_\Omega) -  \phi (\partial_i
\xi^i) \chi_{\Omega} dx \\
&= - \int_{\Re^n} \partial_i \phi \xi^i \chi_{\Omega} - \phi(\partial_i
\xi^i) \chi_{\Omega} dx \\
&= - \int_{\partial_\Omega} \phi n^i dS - \int_{\Re^n} \phi (\partial_i
\xi^i) \chi_{\Omega} dx .
\end{align*}
Specializing to the case $\xi^k \partial_k = \partial_i$ gives \eqref{eq:dchi}.
From \eqref{eq:dchi} we have immediately, by the chain rule and the
differentiation formula for convolutions,
\begin{equation}\label{eq:dVV}
\partial_i \VV[f] = \VV[\partial_i f] - \SS[f n^i] .
\end{equation}
We can now prove the following.
\begin{lemma} \label{lem:VV}
Let $s \geq 0$, $s$ integer.  
Then
$\VV$ defines a continuous map
$$
\VV: H^s(\Omega) \to H^{s+2}(\Omega).
$$
\end{lemma}
\begin{proof}
The case $s=0$ follows from the standard interior estimate for the Poisson
equation, cf. \cite[Chapitre 2, Th\'eor\`eme 3.1]{lions:magenes:vol1}.
The proof now proceeds by induction, with $s=0$ as base.
Let $s \geq 1$ and suppose we have proved
the statement for $s-1$. Let $f \in H^s(\Omega)$. By the trace theorem,
\cite[Chapter 4, Proposition 4.5]{taylor:PDE:I:MR2744150},
$\Tr_{\partial\Omega}[fn^i] \in H^{s-1/2}(\partial\Omega)$ and by lemma
\ref{lem:SS}, $\SS[f n^i] \in H^{s+1/2}(\partial\Omega)$. Now, $\SS[f n^i]$
is harmonic in $\Omega$ with trace on $\partial\Omega$ in
$H^{s+1/2}$. It follows that $\SS[f n^i] \in H^{s+1}(\Omega)$.
Further, by the induction hypothesis, $\VV[\partial_i f] \in
H^{s+1}(\Omega)$, and hence $\partial_i\VV[f] \in H^{s+1}(\Omega)$. Again by
the induction hypothesis, $\VV[f] \in H^{s+1}(\Omega)$, and hence it follows
that $\VV[f] \in H^{s+2}(\Omega)$, which closes the induction and gives the result.
\end{proof}


\providecommand{\bysame}{\leavevmode\hbox to3em{\hrulefill}\thinspace}
\providecommand{\MR}{\relax\ifhmode\unskip\space\fi MR }
\providecommand{\MRhref}[2]{%
  \href{http://www.ams.org/mathscinet-getitem?mr=#1}{#2}
}
\providecommand{\href}[2]{#2}

\end{document}